\pdfoutput=1 
\documentclass[%
a4paper,
UKenglish,%
autoref]{lipics-v2019}
\title{How Reversibility Can Solve Traditional Questions: The Example of %
Hereditary History-Preserving Bisimulation}%

\titlerunning{How Reversibility can Solve Traditional Questions}
\author{Clément Aubert}{School of Computer and Cyber Sciences, Augusta University, USA}{caubert@augusta.edu}{https://orcid.org/0000-0001-6346-3043}{}
\author{Ioana Cristescu}{Tarides, Paris}{ioana@tarides.com}{}{}
\authorrunning{C. Aubert and I. Cristescu}
\Copyright{Clément Aubert and Ioana Cristescu}
\category{}%
\relatedversion{}%
\supplement{}%
\acknowledgements{The authors would like to thank John Natale for correcting the definition of postfixing and the reviewers of an earlier version of this work.%
}%

\EventEditors{John Q. Open and Joan R. Access}
\EventNoEds{2}
\EventLongTitle{42nd Conference on Very Important Topics (CVIT 2016)}
\EventShortTitle{CVIT 2016}
\EventAcronym{CVIT}
\EventYear{2016}
\EventDate{December 24--27, 2016}
\EventLocation{Little Whinging, United Kingdom}
\EventLogo{}
\SeriesVolume{42}
\ArticleNo{23}

\usepackage{mathtools}
\usepackage{ebproof}

\usepackage{thm-restate}
\usepackage[inline]{enumitem} %
\usepackage{stmaryrd}

\usepackage{cases}
\usepackage{empheq}

\usepackage[OMLmathsfit]{isomath}

\usepackage{multicol}

\usepackage{xspace}
\usepackage{csquotes}

\usepackage{tikz}
\usetikzlibrary{%
	arrows.meta, %
	calc,
	backgrounds,
	fit,
	patterns}
\usepackage{graphicx}

\usepackage{stackengine,scalerel,graphicx,trimclip,wasysym}

\usepackage{xcolor,
	colortbl}

\usepackage{microtype}

\savestack\zigzagtextstyle{\vphantom{()}\scalebox{.86666}[1]{\kern-.5pt\AC\kern-.5pt}}
\savestack\onelinetextstyle{\vphantom{()}\scalebox{1}[1]{\kern-1pt{$-$}\kern-1pt}}
\savestack\twolinetextstyle{\vphantom{()}\scalebox{1}[1]{\kern-1pt{$=$}\kern-1pt}}
\savestack\ZigZagtextstyle{%
 $\vphantom{()}\smash{\raisebox{-.002em}{$\vcenter{\hbox{%
 \stackengine{-.805em}{\zigzagtextstyle}{\zigzagtextstyle}{O}{c}{F}{F}{S}}}$}}$}
\newlength\repwidth
\repwidth=\wd\zigzagtextstylecontent\relax
\newcommand\rightarrowhead{\clipbox{4pt -2pt 0pt -2pt}{$\rightarrow$}}
\newcommand\Rightarrowhead{\clipbox{4pt -2pt 0pt -2pt}{$\Rightarrow$}}
\newcommand\rrightarrowhead{\clipbox{5.5pt -2pt 0pt -2pt}{$\rightarrow\kern-8.5pt\rightarrow$}}
\newcommand\RRightarrowhead{\clipbox{5.5pt -2pt 0pt -2pt}{$\Rightarrow\kern-8.5pt\Rightarrow$}}
\newcommand\fXarrowtextstyle[2]{$\vphantom{()}\smash{\vcenter{\hbox{\kern-.03\repwidth%
 \clipbox{-.03\repwidth{} 0pt .527\repwidth{} 0pt}{#1}}}}#2$}
\savestack\fzigarrowtextstyle{\fXarrowtextstyle{\zigzagtextstyle}{\rightarrowhead}}
\savestack\fZigarrowtextstyle{\fXarrowtextstyle{\ZigZagtextstyle}{\Rightarrowhead}}
\savestack\flinearrowtextstyle{\fXarrowtextstyle{\onelinetextstyle}{\rightarrowhead}}
\savestack\fLinearrowtextstyle{\fXarrowtextstyle{\twolinetextstyle}{\Rightarrowhead}}
\savestack\fzzigarrowtextstyle{\fXarrowtextstyle{\zigzagtextstyle}{\rrightarrowhead}}
\savestack\fZZigarrowtextstyle{\fXarrowtextstyle{\ZigZagtextstyle}{\RRightarrowhead}}
\savestack\fllinearrowtextstyle{\fXarrowtextstyle{\onelinetextstyle}{\rrightarrowhead}}
\savestack\fLLinearrowtextstyle{\fXarrowtextstyle{\twolinetextstyle}{\RRightarrowhead}}
\newcommand\zigzag{\scalerel*{\zigzagtextstyle}{()}}

\newcommand\oneline{\scalerel*{\onelinetextstyle}{()}}
\newcommand\twoline{\scalerel*{\twolinetextstyle}{()}}
\newcommand\fzigarrow{\scalerel*{\fzigarrowtextstyle}{()}}

\newcommand\flinearrow{\scalerel*{\flinearrowtextstyle}{()}}
\newcommand\fLinearrow{\scalerel*{\fLinearrowtextstyle}{()}}

\newcommand\fLLinearrow{\scalerel*{\fLLinearrowtextstyle}{()}}
\newcommand\zigzagarrow[1][]{\Xarrow{\zigzag}{\fzigarrow}{#1}{-.65}}

\newcommand\linearrow[1][]{\Xarrow{\oneline}{\flinearrow}{#1}{-.65}}
\newcommand\Linearrow[1][]{\Xarrow{\twoline}{\fLinearrow}{#1}{-.40}}

\newcommand\LLinearrow[1][]{\Xarrow{\twoline}{\fLLinearrow}{#1}{-.40}}
\newcommand\Xarrow[4]{\ThisStyle{\mathrel{\Xarrowhelp{#1#2}{#3}{#1}{#4}}}}
\newcommand\Xarrowhelp[4]{%
 \setbox0=\hbox{$\SavedStyle#1$}%
 \setbox2=\hbox{$\SavedStyle_{\,\,#2\,\,}$}%
 \ifdim\wd0<\wd2\relax\Xarrowhelp{#3#1}{#2}{#3}{#4}%
 \else\stackengine{#4\LMex}{\copy0}{\copy2\,}{O}{c}{F}{T}{S}\fi%
}%
\newcommand{\efs}[1][]{\linearrow[#1]}
\newcommand{\ebs}[1][]{\zigzagarrow[#1]}
\newcommand{\efbs}[1][]{\Linearrow[#1]}
\newcommand{\redl}[1]{\efs[#1]}

\newcommand{\revredl}[1]{\ebs[#1]}
\newcommand{\fwlts}[2]{\efs[#1:#2]}
\newcommand{\bwlts}[2]{\ebs[#1:#2]}
\newcommand{\fbwlts}[2]{\efbs[#1:#2]}
\newcommand{\tfbwlts}{\LLinearrow}

\newcommand{\congru}{\equiv}
\DeclareMathOperator{\co}{co}
\newcommand{\HPB}{\textnormal{HPB}\xspace}
\newcommand{\HHPB}{\textnormal{HHPB}\xspace}
\newcommand{\HorHPB}{\textnormal{(H)HPB}\xspace}
\newcommand{\BF}{\textnormal{B\&F%
}\xspace}
\newcommand{\SBF}{\textnormal{SB\&F%
}\xspace}
\newcommand{\FR}{\textnormal{FR%
}\xspace}

\newcommand{\emptymem}{\emptyset}
\newcommand{\orig}[1]{O_{#1}} \newcommand{\fork}{\curlyvee}
\newcommand{\bs}{\backslash}
\newcommand{\cat}[1]{\mathbb{#1}}
\newcommand{\catcs}{\cat{C}}
\newcommand{\catics}{\cat{D}}
\DeclareMathOperator{\ids}{\mathsf{I}}
\newcommand{\card}[1]{|#1|}
\DeclareMathOperator{\id}{id}

\newcommand{\iden}{\mathsfit{m}}
\newcommand{\iconf}{\ensuremath{\mathcal{I}}}
\DeclarePairedDelimiter{\encc}{\llbracket}{\rrbracket}
\newcommand{\enc}{\encc}
\newcommand{\nm}[1]{\ensuremath{\mathrm{nm}(#1)}}
\newcommand{\names}{\ensuremath{\mathsf{N}}}

\newcommand{\rproc}{\ensuremath{\mathsf{R}}}
\newcommand{\labels}{\ensuremath{\mathsf{L}}}

\newcommand{\iso}{\ensuremath{\cong}}
\newcommand{\out}[1]{\overline{#1}}
\newcommand{\power}{\ensuremath{\mathcal{P}}}
\newcommand{\conf}{\ensuremath{\mathcal{C}}}
\newcommand{\rel}{\ensuremath{\mathcal{R}}}
\DeclarePairedDelimiter{\mem}{\langle}{\rangle}
\newcommand{\labl}{\ell}

\newcommand{\frestr}[1]{\mathord{\upharpoonright_{#1}}}
\newcommand{\crestr}[1]{\mathord{\upharpoonleft_{#1}}}
\newcommand{\confzero}{\mathbf{0}}
\newcommand{\iconfzero}{\mathbf{0}}
\newcommand{\nat}{\mathbb{N}}
\newcommand{\integer}{\nat}
\newcommand{\setst}{ \mid }

\newcommand{\BNFsepa}{\enspace \Arrowvert \enspace}
\newcommand{\functoric}{\ensuremath{\mathcal{F}}}
\newcommand{\functorci}{\ensuremath{\mathcal{S}}}
\DeclarePairedDelimiter{\encm}{\lfloor}{\rfloor}
\DeclarePairedDelimiter{\encmp}{\lceil}{\rceil}
\newcommand{\namelist}[1]{\overrightarrow{#1}}

\AtBeginDocument{

	\def\itemautorefname~#1\null{\textcolor{darkgray}{\sffamily\bfseries\upshape\mathversion{bold}#1.}\null}
}

\DeclareMathOperator{\post}{{::}} %

\newcommand*{\eg}{e.g.\@\xspace}
\newcommand*{\cf}{cf.\@\xspace}
\newcommand*{\ie}{i.e.\@\xspace}

\makeatletter
\newcommand*{\etc}{%
 \@ifnextchar{.}%
 {etc}%
 {etc.\@\xspace}%
}
\makeatother

\renewcommand{\st}{s.t.\@\xspace}
\newcommand{\withoutlog}{without loss of generality\xspace}
\newcommand*{\resp}{resp.\@\xspace}
\newcommand*{\wrt}{w.r.t.\@\xspace}

\newcommand{\lop}{l\&o-p\xspace}

\newcommand\hypo{\Hypo}
\newcommand\infer{\Infer}

\renewcommand{\footnotesize}{\fontsize{8pt}{10pt}\selectfont} %

\nolinenumbers
\makeatletter%
\begin{document}

\hideLIPIcs
\maketitle
\begin{abstract}
	Reversible computation opens up the possibility of overcoming some of the hardware's current physical limitations.
	It also offers theoretical insights, as it enriches multiple paradigms and models of computation, and sometimes retrospectively enlightens them.
	Concurrent reversible computation, for instance, offered interesting extensions to the Calculus of Communicating Systems, but was still lacking a natural and pertinent bisimulation to study processes equivalences.
	Our paper formulates an equivalence exploiting the two aspects of reversibility: backward moves and memory mechanisms.
	This bisimulation captures classical equivalences relations for denotational models of concurrency (History- and hereditary history-preserving bisimulation, \HorHPB), that were up to now only partially characterized by process algebras.
	This result gives an insight on the expressiveness of reversibility, as both backward moves and a memory mechanism---providing \enquote{backward determinism}---are needed to capture \HHPB.

\end{abstract}

\section{Introduction}
\label{sec:intro}

\subparagraph*{The Benefits of Reversible Computation}
Future progresses in computing may heavily rely on reversibility~\cite{Frank2017}.
The foreseeable limitations of conventional semiconductor technology, Lauder's principle~\cite{PhysRevLett.113.190601}---promising low-energy consumption for reversible computers---and quantum computing~\cite{Nielsen2010}---intrinsically reversible and now within reach~\cite{Arute2019}---motivated a colossal push toward a better understanding of reversible computation.
Those efforts have given birth to new paradigms~\cite{Frank17}, richer models of computation (\eg for automata~\cite{Holzer2017}, Petri nets~\cite{Frutos-EscrigK2019, Philippou2018}, Turing machines~\cite{Axelsen2016}) and richer semantics, sometimes for preexisting calculus like the Calculus of Communicating Systems (CCS) and the \(\pi\)-calculus~\cite{Cristescu2013,Cristescu2015b}.
Those new perspectives sometimes additionally give in retrospect a better understanding of \enquote{traditional} (\ie, irreversible) computation, %
and our contribution illustrates this latter aspect.

\subparagraph*{Summary}
In brief terms, we offer a solution to an open problem on classes of equivalence of (non-reversible) concurrent processes thanks to reversibility.
\emph{Hereditary history-preserving bisimulation} (\HHPB) is considered \enquote{the gold standard} for establishing equivalence classes on \enquote{true} models of concurrency~\cite{Joyal1996b, Glabeek1989}.
However, no relation expressed in syntactical terms (\eg on CCS) was known to capture them, despite intensive efforts: previous results~\cite{Aubert2016jlamp, Phillips2007} characterized \HHPB on limited classes of processes, that excluded CCS terms as simple as \(a.a\), \(a + a\), \(a\mid a\) or containing similar patterns.
We prove in this paper a somewhat expected result, namely that adding mechanisms to reverse the computation and keep track of the past, thus making CCS \enquote{backwards deterministic}~\cite{Bednarczyk1991}, enables syntactical characterizations of \HHPB. %
This result uses natural reformulations of canonical CCS bisimulations, provides a \enquote{meaningful} bisimulation for reversible calculi, and furthermore validates the mechanism we use to reverse the computation in a concurrent set-up.
It also connects with previous semantics of reversible processes~\cite{Aubert2016jlamp,Graversen2018} and supports categorical treatment.

\subparagraph*{Reversing Concurrent Computation}
Reversible systems can backtrack %
and return
to a %
previous state.
Implementing reversibility %
often requires a mechanism to record the history of the execution.
Ideally, this history should be \emph{complete}, so that every forward step can be backtracked, and \emph{minimal}, so that only the relevant information is saved.
Concurrent programming languages have additional requirements: the history should be \emph{distributed}, to avoid centralization, and should prevent steps that required a synchronization with other parts of the program to backtrack without undoing this synchronization.
To fulfill those requirements, Reversible CCS (RCCS)~\cite{Danos2004,Danos2005} uses \emph{memories} attached to the threads of a process, and CCS with Communication Keys (CCSK)~\cite{Phillips2007b} \enquote{marks} each occurrence of an action.
The two calculi are equi-expressive~\cite{Lanese2019} and are conservative extensions over CCS, and in this paper we will use RCCS, to benefit from its previous operational semantics~\cite{Aubert2016jlamp}.

\subparagraph*{Equivalences for Denotational Models of Concurrency}
A theory of concurrent computation, be it reversible or irreversible, relies not only on a syntax and a model, but also on \enquote{meaningful} behavioral equivalences.
This paper
defines new bisimulations on RCCS, and proves their relevance by connecting them to
bisimulations on configuration structures~\cite{Winskel1986}, a classical denotational model for concurrency.
In configuration structures, an \emph{event} represents an execution step, and a \emph{configuration}---a set of events that occurred---represents a state.
A forward transition is then represented as moving from a configuration to one of its supersets.
Backward transitions have a \enquote{built-in} representation: it suffices to move from a configuration to one of its subsets.
Among the behavioral equivalences on configuration structures,
some of them, like \emph{history-} and \emph{hereditary history-preserving bisimulation} (\HPB and \HHPB)~\cite{Baldan2014, Bednarczyk1991, Froschle1999, Phillips2012, Phillips2014, Glabbek1996}, use that \enquote{built-in} notion of reversibility.
\HPB and \HHPB are usually regarded as the \enquote{canonical}~\cite[p.~94]{Phillips2007} or \enquote{strongest desirable true}~\cite[p.~2]{Phillips2014} concurrency equivalences
because they preserve causality, branching, their interplay, and are the coarser (or finer, for \HPB) reasonable equivalences with these properties~\cite[p.~309]{Glabbeek2001}.
\HHPB also naturally equal path bisimulations~\cite{Joyal1996b}, an elegant notion of equivalence relying on category theory, and can be captured concisely using event identifier logic~\cite{Phillips2014}.
However, no relation on labeled transition systems (LTS), being CCS, RCCS, or CCSK, was known to capture it on all processes.

\subparagraph*{Encoding Reversible Processes in Configuration Structures}

There have been multiple attempts%
~\cite{Aubert2016jlamp,Phillips2007b} to transfer equivalences defined on denotational models, by construction adapted for reversibility, into reversible process algebras.
Showing that an equivalence on configuration structures corresponds to one on processes starts by defining an encoding of the latter in the former. %
Encodings, for RCCS~\cite{Aubert2016jlamp} and CCSK~\cite{Graversen2018} alike, generally consider only \emph{reachable} reversible processes, and map them to one particular configuration in the encoding of its \emph{origin}, the CCS \enquote{memory-less} process to which it can backtrack.
We come back to the relation between this encoding and our current work in \hyperref[sec:conclusion]{conclusion}.

\subparagraph*{Contribution}
This paper improves on previous results~\cite{Aubert2016jlamp,Phillips2007} %
by defining relations on syntactical terms that correspond to \HHPB on \emph{all} processes, and hence are proven to be \enquote{the right} bisimulation to study reversible computation.
This result relies on encoding the processes' memories into \emph{identified configuration structures}, an extension to configuration structures that constitutes our second contribution.%

RCCS is exploited as a syntactic tool to decide \HHPB for CCS processes, by endowing them with
\begin{enumerate*}
	\item backtracking capabilities and
	\item a memory mechanism%
\end{enumerate*}%
.
This gives a precious insight on the expressiveness of reversibility, as we show that having only one of those tools %
is not enough to define relations that capture \HHPB.
The memories attached to a process are no longer only a syntactic layer to implement reversibility, but become essential to define and capture equivalences, thanks to the \enquote{backward determinism} they provide.

Recursion is not treated in this paper: its treatment amounts to unfolding processes and structures up to a certain level, and it strongly suggests that there are not much insight to gain from its development, that we reserve as a technical appendix for future work.

\subparagraph*{Related Work}
The correspondences between \HHPB and back-and-forth bisimulations for restricted classes of processes~\cite{Aubert2016jlamp,Phillips2007} inspired some of the work presented here.
This correspondence has been studied in the denotational world, on structures without auto-concurrency~\cite{Bednarczyk1991}.

The insight that \enquote{backward determinism} is essential to capture \HHPB is also used in the event identifier logic~\cite{Phillips2014} and in one if its sub-system~\cite{Baldan2014a}, that both
characterize %
\HHPB %
without restricting the structures considered.
As in our case, the authors exploits %
identifiers to eliminate constraints due to label \enquote{duplication}. %
Our work is similar, taking place in the operational world instead of the logical one.
Note that we do not introduce extra syntax constructions on the LTS, but simply use the ones provided by RCCS, that we use \enquote{as is}.

In the operational semantics, the previous attempts to characterize \HHPB wrongly focused on the backward capabilities of processes and considered the memory mechanisms only as a tool to achieve it.
As a result, those bisimulations %
could not decide that the encoding of %
\((a.a) \mid b\) and \(a \mid a \mid b\) were not \HHPB.
Instead, their characterizations of \HHPB were applicable only on \enquote{non-repeating}~\cite{Phillips2007} or \enquote{singly-labeled}~\cite{Aubert2016jlamp} processes. %
By integrating the memory mechanism, the equivalence relation we introduce and consider can correctly determine that these two processes are not \HHPB, as we detail in \autoref{ex:not_hhpb}.

Finally, our approach shares similarity with causal trees---in the sense that only \emph{part of the execution}, its \enquote{past}, is encoded in a denotational representation---which were used to characterize \HPB for CCS terms~\cite{Darondeau1990a}.
Capturing \HorHPB with novel techniques can also impact model-checking and decidability issues~\cite{Baldan2014a}, but we leave this as future work.

\section{Denotational and Operational Models for the Reversible CCS}
\label{sec:prelim}

We write \(\subseteq\) the set inclusion, \(\mathcal{P}\) the power set, \(\bs\) the set difference, %
\(f: A \to B\) (\resp \(f: A \rightharpoonup B\)) for the (\resp partial) function from \(A\) to \(B\),
\(f \frestr{C}\) the restriction of \(f\) to \(C \subseteq A\), and \(f \cup \{a \mapsto b\}\) the function defined as \(f\) on \(A\) that additionally maps \(a \notin A\) to \(b\).

\subsection{Identified Configuration Structures}
\label{sec:cs}

Labeled configuration structures~\cite{Glabbeek2009,Winskel1982} are a classical non-interleaving model of concurrent computation, that we enrich here with identifiers.
We then show that they support categorical understanding and the usual operations just as well as their \enquote{un-identified} variations.

\begin{definition}[(Identified) Configuration structure]
	\label{def:conf_str}
	A \emph{configuration structure} \(\conf\) is a tuple \((E, C, L, \labl)\) where \(E\) is a set of events, \(L\) is a set of \emph{labels}, \(\labl : E \to L \) is a labeling function and \(C \subseteq \power(E)\) is a set of subsets satisfying:
	\begin{align*}
		 & \forall x \in C, \forall e \in x, \exists z \in C\text{ finite}, e\in z, z \subseteq x \tag{Finiteness}\label{Finiteness} \\
		 & \forall x \in C, \forall d, e \in x, d \neq e\Rightarrow
		\exists z \in C, z \subseteq x, d\in z \iff e\notin z\tag{Coincidence Freeness}\label{CoincidenceFreeness}                   \\
		 & \forall X \subseteq C, \exists y\in C \text{ finite}, \forall x \in X, x\subseteq y \Rightarrow {\textstyle \bigcup }     %
		X \in C \tag{Finite Completeness}\label{FiniteCompleteness}                                                                  \\
		 & \forall x,y\in C, x\cup y\in C \Rightarrow x\cap y\in C \tag{Stability}\label{Stability}
		\shortintertext{If \(\conf\) also has %
			a set of \emph{identifiers} \(I\) and an \emph{identifying function} \(\iden : E \to I\) satisfying:}
		 & \forall x \in C, \forall e_1, e_2 \in x, \iden(e_1) \neq \iden(e_2) \tag{Collision Freeness}\label{eq:collision}
	\end{align*}
	then we write \(\conf \oplus \iden\), %
	and say that \(\iconf = (E, C, L, \labl, I, \iden)\) is an \emph{identified configuration structure}, or \emph{\(\iconf\)-structure}, and call \(\conf\) \emph{the underlying configuration structure} of \(\iconf\).

	We denote with \(\confzero\) both the \(\iconf\)-structure and its underlying configuration structure with \(C = \{\emptyset\}\), and, for \(x, y \in C\), we write \(x \redl{e} y\) and \(y \revredl{e}x\) if \(x=y\cup\{e\}\).
\end{definition}

We omit the identifiers when representing \(\iconf\)-structures %
and write the label for the event (with a subscript if multiple events shares a label).

\begin{definition}[Causality, concurrency, and maximality]
	\label{def:causality}
	For all \(\iconf\), \(x \in C\) and \(d, e \in x\), the \emph{causality relation on \(x\)} is given by \(d <_{x} e\) iff \(d \leqslant_x e\) and \(d \neq e\), where \(d \leqslant_x e\) iff for all \(y \in C\) with \(y \subseteq x\), we have \(e \in y \Rightarrow d \in y\).
	The \emph{concurrency relation on \(x\)} %
	is given by \(d \co_x e\) iff \(\neg (d <_x e \vee e <_x d) \).
	Finally, \(x\) is \emph{a maximal configuration in \(\iconf\)} if \(\forall y \in C\), \(x=y\) or \(x\not\subseteq y\).
\end{definition}

\begin{figure}
	{\centering
		\begin{minipage}[b]{.29\linewidth}
			\begin{tikzpicture}\node (emptyset) at (0,-.9) {\(\emptyset\)};
				\node (a1) at (-1,0) {\(\{a_1\}\)};
				\node (a2) at (1,0) {\(\{a_2\}\)};
				\draw [->] (emptyset) -- (a1);
				\draw [->] (emptyset) -- (a2);
			\end{tikzpicture}
			\subcaption{%
			}\label{fig:ex:3}
		\end{minipage}
		\hfill
		\begin{minipage}[b]{.29\linewidth}
			\begin{tikzpicture}\node (emptyset) at (0,-.9) {\(\emptyset\)};
				\node (a1) at (-1,0) {\(\{a_1\}\)};
				\node (a2) at (1,0) {\(\{a_2\}\)};
				\node (a1a2) at (0, 1) {\(\{a_1, a_2\}\)};
				\draw [->] (emptyset) -- (a1);
				\draw [->] (emptyset) -- (a2);
				\draw [->] (a1) -- (a1a2);
				\draw [->] (a2) -- (a1a2);
			\end{tikzpicture}
			\subcaption{%
			}\label{fig:ex:2}
		\end{minipage}
		\hfill
		\begin{minipage}[b]{.4\linewidth}
			\begin{tikzpicture}
				\node (emptyset) at (1,-.9) {\(\emptyset\)};
				\node (b) at (-1,0) {\(\{a\}\)};
				\node (c) at (1,0) {\(\{\out{a}\}\)};
				\node (bcp) at (0, .9) {\(\{a, \out{a}\}\)};
				\node (ba) at (-1.5,.9) {\(\{a, b\}\)};
				\node[align=center] (bapc) at (-0.75,1.8) {\(\{a, \out{a}, b\}\)};
				\draw [->] (emptyset) -- (b);
				\draw [->] (emptyset) -- (c);
				\draw [->] (b) -- (bcp);
				\draw [->] (c) -- (bcp);
				\draw [->] (b) -- (ba);
				\draw [->] (ba) -- (bapc);
				\draw [->] (bcp) -- (bapc);
				\node (t) at (2, 0) {\(\{\tau\}\)};
				\node (tb) at (2, 1) {\(\{\tau, b\}\)};
				\draw [->] (emptyset) to (t);
				\draw [->] (t) to (tb);
			\end{tikzpicture}
			\subcaption{%
			}\label{fig:ex:4}
		\end{minipage}
		\caption{Examples of \(\iconf\)-structures
		}
		\label{fig:ex}
	}
\end{figure}

\begin{example}
	Consider \autoref{fig:ex}, where we let the events have distinct arbitrary identifiers: two events with complement names as labels can happen at the same time (\autoref{fig:ex:4}), in which case they are labeled with \(\tau\),
	as is usual in CCS (\autoref{sec:RCCS}). %
	In \autoref{fig:ex:4}, \(a <_{\{a,b\}} b\), \(a <_{\{a,\out a,b\}} b\) and \(\tau <_{\{\tau, b\}} b\); and in \autoref{fig:ex:2}, \(a_1 \co_{\{a_1, a_2\}} a_2\).
	An \(\iconf\)-structure can have one (\autoref{fig:ex:2}) or multiple (Fig.~\ref{fig:ex:3} and \ref{fig:ex:4}) maximal configurations.
\end{example}

\subparagraph*{Categorical point of view}

We remind in \autoref{sec:ct-app} that configuration structures and \enquote{structure-preserving} functions %
form a category. %
We also prove that a similar category %
can be defined with \(\iconf\)-structures as objects and define a forgetful functor %
that returns the underlying configuration structure.
This development supports the interest and validity of studying \(\iconf\)-structures, but
can be omitted, except for the notion of morphisms:

\begin{definition}[Morphism of \(\iconf\)-structure]
	A morphism \(f = (f_E, f_L, f_C, f_{\iden})%
	\) between \(\iconf_1\) and \(\iconf_2\) is given by \(f_E : E_1 \to E_2\) such that \(\labl_2(f_E(e)) = f_L(\labl_1(e))\), for \(f_L : L_1\to L_2\); \(f_C : C_1 \to C_2\) defined as \(f_C(x) = \{ f_E(e) \setst e\in x\}\), and \(f_{\iden} : I_1 \to I_2\) such that \(f_{\iden}(\iden_1(e)) = \iden_2(f_E(e))\).
	We often write \(f\) for all the components, and write \(\iconf_1 \iso \iconf_2\) if \(f\) is an isomorphism.
\end{definition}

\subparagraph*{Operations on \texorpdfstring{\(\iconf\)}{I}-configurations}
\label{sec:op-def}

Operations on \(\iconf\)-structures are conservative extensions over their counterparts on configuration structures---forgetting about event identifiers gives back the \enquote{un-identified} definition~\cite{Winskel1982}---, except for the parallel composition.
The intuition here is that configuration structures encode CCS processes, and \(\iconf\)-structures encode memories of RCCS processes%
, where parallel composition has a different meaning.
Examples of those operations will be given in \autoref{sec:examples}, after introducing the calculi in \autoref{sec:RCCS} and the encodings in \autoref{sec:css-enc}.
\autoref{sec:op-proof} is devoted to proving the correctness of those operations.

Given two sets \(A\), \(B\), and a symbol \(\star \notin A \cup B\) denoting \emph{undefined}, we write \(C^{\star}= C \cup \{\star\}\) if \(\star \notin C\) and define the partial product~\cite[Appendix A]{Winskel1982} of \(A\) and \(B\) to be %
\[
	A \times_{\star} B =
	\{(a,\star)\mid a \in A\}\cup\{(\star, b)\mid b \in B\}
	\cup\{(a, b)\mid a \in A, b\in B\}
\]
and the two projections to be \(\pi_1: A \times_{\star} B \to A^{\star}\) and \(\pi_2 : A \times_{\star} B \to B^{\star}\).

\begin{definition}[Operations on \(\iconf\)-structures]
	Given \(\iconf_i = (E_i, C_i, L_i, \labl_i, I_i, \iden_i)\), for \(i = 1, 2\),
	\label{icat-op-def}
	\begin{description}
		\item[The relabeling] of \(\iconf_1\) along \(\labl': E_1 \to L\) is \(\iconf_1 [\labl'/\labl_1] = (E_1, C_1, L, \labl', I_1,\iden_1)\).
		\item[The reidentifiying] of \(\iconf_1\) along \(\iden : E_1 \to I\) is \(\iconf_1 [\iden/\iden_1] = (E_1, C_1, L_1, \labl_1, I,\iden)\), provided \(\iconf_1 [\iden/\iden_1]\) respects the \ref{eq:collision} condition.
		\item[The restriction of a set of events] \(A \subseteq E_1\) in \(\iconf_1\) is
		      \(\iconf_1\crestr{A} = (E, C, L, \labl_1\frestr{E}, I, \iden_1\frestr{E})\)
		      , where
		      \(E =E_1{\setminus{A}}\),
		      \(C = \{x \setst x \in C_1 \text{ and } x \cap A = \emptyset\}\),
		      \(L = \{ a \setst \exists e \in E_1{\setminus{A}}, \labl_1(e) = a\}\) and \(I = \{ i \setst \exists e \in E_1{\setminus{A}}, \iden_1(e) = i\}\).
		\item[The restriction of a set of labels] \(L \subseteq L_1\) in \(\iconf_1\) is \(\iconf_1\crestr{L} = \iconf_1 \crestr{E_1^{L}}\) where \(E_1^{L}=\{e \in E_1 \setst \labl_1(e) \in L\}\).
		      We write \(\iconf_1\crestr{a}\), when the restricting set of labels \(L\) is the singleton \(\{a\}\).

		\item[The prefixing] of \(\iconf_1\) by the label \(a\) is \(a . \iconf_1 = (E, C, L,\labl, I,\iden)\) where
			      {
				      \begin{multicols}{2}
					      \begin{itemize}
						      \item \(E = E_1 \cup \{e\}\), for \(e \notin E_1\),
						      \item \(C = \{ x \setst x = \emptyset \text{ or } \exists x' \in C_1, x = x' \cup e\}\),
						      \item \(L = L_1\cup\{a\}\),
						      \item \(\labl = \labl_1 \cup \{e \mapsto a\}\),
						      \item \(I = I_1\cup\{i\}\), for \(i \notin I_1\)
						      \item \(\iden = \iden_1 \cup \{e \mapsto i\}\).
					      \end{itemize}
				      \end{multicols}
			      }
		\item[The postfixing] of \((a, i)\) to \(\iconf_1\) is defined if \(i \notin I_1\) as \(\iconf_1 \post (a, i) = (E, C, L,\labl, I,\iden)\) where everything is as in \(a . \iconf_1\), except that \(C = C_1 \cup \{x \cup \{e\} \setst x \in C_1 \text{ is maximal and finite}\}\).
		\item[The nondeterministic choice] of \(\iconf_1\) and \(\iconf_2\) is \(\iconf_1+\iconf_2 = (E, C, L,\labl, I,\iden)\), where
			      {
				      \setlength\multicolsep{\topsep}
				      \begin{multicols}{2}
					      \begin{itemize}
						      \item \(E=\{\{1\}\times E_1\}\cup\{\{2\}\times E_2\}\) with \(\pi_1 : E \to \{1,2\}\) and \(\pi_2: E \to E_1\cup E_2\),
						      \item \(C=\{\{i\}\times x \setst x\in C_i\}\),
						      \item \(L = L_1 \cup L_2\),
						      \item \(\labl(e)=\labl_i(\pi_2(e)\)) for \(\pi_1(e) = i\),
						      \item \(I =I_1\cup I_2 \),%
						      \item \(\iden(e) = \iden_i(\pi_2(e))\) for \(\pi_1(e) = i\).
					      \end{itemize}
				      \end{multicols}
			      }
		\item[The product] of \(\iconf_1\) and \(\iconf_2\) is \(\iconf_1\times\iconf_2 = (E, C, L, \labl, I, \iden)\), where:

		      \begin{itemize}
			      \item \(E = E_1 \times_{\star} E_2\), with \(\pi_i:E\to E_i^{\star}\) the projections of the partial product,

			      \item For \(i \in \{1, 2\}\), define the projections \(\gamma_i : \iconf_1 \times \iconf_2 \to \iconf_i\) and the configurations \(x\in C\): %
			            \begin{gather*}
				            \forall e \in E, \gamma_i(e)=\pi_i(e), \gamma_i(\labl_i(e)) = \labl_i(\pi_i(e)), \gamma_i(\iden_i(e)) = \iden_i(\pi_i(e))\\
				            \gamma_i(x)\in C_i, \text{ with } \gamma_i(x) = \{e_i \setst \pi_i(e) = e_i \neq \star \text{ and } e \in x\}\\
				            \forall e,e'\in x, \pi_1(e)=\pi_1(e')\neq\star\text{ or }\pi_2(e)=\pi_2(e')\neq\star \Rightarrow e=e' \\
				            \forall e \in x, \exists z\subseteq x \text{ finite},\gamma_i(z) \in C_i, e\in z \\
				            \forall e, e' \in x, e\neq e'\Rightarrow \exists z\subseteq x, \gamma_i(z) \in C_i,e\in z \iff e'\notin z
			            \end{gather*}
			      \item \(\labl:E \to L = L_1\times_{\star} L_2\) is \(
			            \labl(e) = \begin{dcases*}
				            (\labl_1(e_1), \star) & if \(\pi_2(e)=\star\)\\
				            (\star, \labl_2(e_2)) & if \(\pi_1(e)=\star\) \\
				            (\labl_1(e_1),\labl_2(e_2)) & otherwise
			            \end{dcases*}
			            \)

			      \item \(\iden:E \to I= I_1\times_{\star} I_2 \) is \(
			            \iden(e) = \begin{dcases*}
				            (\iden_1(\pi_1(e)),\star) & if \(\pi_2(e)=\star\)\\
				            (\star,\iden_2(\pi_2(e))) & if \(\pi_1(e)=\star\)\\
				            (\iden_1(\pi_1(e)),\iden_2(\pi_2(e))) & otherwise
			            \end{dcases*}
			            \)

		      \end{itemize}
	\end{description}
\end{definition}

We now recall the definition of parallel composition for configuration structure, and detail the definition for \(\iconf\)-structures.
Parallel composition consists of a combination of product, relabelling, reidentifiying for the \(\iconf\)-structures, and restriction.
For the relabelling operation, we use a \emph{synchronization algebra}~\cite{Winskel1995} \((S,\star, \bot)\) consisting of a commutative and associative operation \(\bullet\) on a set of labels \(S \cup \{\star,\bot\}\), where \(\{\star,\bot\} \notin S\) and such that \(a \bullet \star = a\) (\eg \(\star\) is the identity element) and \(a \bullet \bot = \bot\) (\eg \(\bot\) is the zero element), for all \(a\in S\). To avoid repetition, below we assume given \((S, \star, \bot)\), such that \(S \subseteq L_1\cup L_2\). We give examples of synchronization algebras in \autoref{sec:css-enc}. %

\begin{definition}[Parallel composition of configuration structures]
	\label{def:par-comp-ccs}

	\emph{The parallel composition} of \(\conf_1\) and \(\conf_2\) is \(\conf_1 \mid_{S} \conf_2 = \big((\conf_1 \times \conf_2) [\labl'/\labl]\big)\crestr{\bot}\) where \(\labl : E_1 \times_{\star} E_2 \to L_1 \cup L_2\) is the labeling function from the product, and \(\labl': E_1 \times_{\star} E_2 \to L_1\cup L_2\cup\{\bot\} \) is defined as \(\labl'(e) = \labl_1(e_1) \bullet \labl_2(e_2)\).

\end{definition}

\begin{definition}[Parallel composition of \(\iconf\)-structures]
	\label{def:par-ics}
	\emph{The parallel composition} of \(\iconf_1\) and \(\iconf_2\), is \(\iconf_1 \mid_{S} \iconf_2 = (\iconf_3 [\iden'/\iden_3] [\labl'/\labl_3])\crestr{\bot}\) where \(\iconf_3 = (E_3, C_3, L_3,\labl_3, I_3,\iden_3)\) is %
	\(\iconf_1\times \iconf_2\), and
	\begin{itemize}
		\item \(\iden' : E_3 \to I_1 \cup I_2 \cup \{\bot_{k} \setst k \in I_1 \times_{\star} I_2\}\) is defined as follows, for \(i \neq \star\):
		      \begin{empheq}[left={\iden'(e)=\empheqlbrace}]{align*}
			      i & \quad \text{if } \iden_3(e) = (i,i) \tag{Sync.\ or Fork}\label{valid-synch}\\
			      i & \quad \text{if }\iden_3(e)=(i, \star) \wedge \forall e_2 \in E_2, \iden_2(e_2) \neq i \tag{Extra \(\star\).\@ 1}\label{extra1}\\
			      i & \quad \text{if }\iden_3(e)=(\star, i) \wedge \forall e_1 \in E_1, \iden_1(e_1) \neq i \tag{Extra \(\star\).\@ 2}\label{extra2}\\
			      \bot_{k} & \quad \text{otherwise, with \(\iden_3(e) = k\) } \tag{Error}\label{error}
		      \end{empheq}
		\item \(\labl' : E_3 \to L_1\cup L_2 \cup \{\bot\}\) maps \(e\) to \(\bot \) if \(\iden'(e) = \bot_k\), and to \(\labl_1(e_1) \bullet \labl_2(e_2)\) otherwise.
	\end{itemize}
\end{definition}

Parallel composition removes from the product the pairs of events that represent \enquote{non-realizable} interactions: for configuration structures, pairs of events that do not represent \emph{possible and future synchronizations}; for \(\iconf\)-structures, pairs of events that do not represent \emph{past synchronizations or forks}. %
Definitions \ref{def:encoding-CCS-cf} and \ref{def:encode_m} will detail how those operations are used to encode a process or a memory, and both types of parallel compositions will be illustrated in Examples \ref{ex:proc-enc-1} and \ref{ex:proc}.

\subsection{Concurrent Communicating Calculi}
\label{sec:RCCS}

Let \(\ids = \integer\) be a set of \emph{identifiers} and \(i, j\) range over it.
Let \(\names=\{a,b,c,\dots\}\) be a set of \emph{names} and \(\out{\names}=\{\out{a},\out{b},\out{c},\dots\}\) its set of \emph{co-names}.
We define the set of labels \(\labels = \names \cup \out{\names} \cup\{\tau\}\), %
and use \(\alpha\) (\resp \(\lambda, \mu, \nu\)) to range over \(\labels\) (\resp \(\labels \bs \{\tau\}\)). The \emph{complement} of a name is given by a bijection \(\out{\cdot}:\names \to \out{\names}\), whose inverse is also written \(\out{\cdot}\), and that we extend to \(\tau\), \ie \(\out{\tau} = \tau\).

\begin{definition}[RCCS Processes]
	The set of \emph{reversible processes \(\rproc\)} is built on top of the set of CCS processes %
	by adding memory stacks to the threads:
	\begin{align}
		P,Q  & \coloneqq P\mid Q \BNFsepa {\textstyle\sum}_{i \geqslant 0} \lambda_i . P_i \BNFsepa P \bs a %
		\tag{CCS Processes}                                                                                 \\
		e    & \coloneqq \mem{i, \lambda, P} \tag{Memory Events}                                            \\
		m    & \coloneqq \emptymem \BNFsepa \fork. m \BNFsepa e. m \tag{Memory Stacks}                      \\
		T    & \coloneqq m \rhd P \tag{Reversible Threads}                                                  \\
		R, S & \coloneqq T \BNFsepa R \mid S \BNFsepa R\bs a \tag{RCCS Processes}
	\end{align}
	We denote \(\ids(e)\) (resp \(\ids(m)\), \(\ids(R)\)) the set of identifiers occurring in \(e\) (\resp \(m\), \(R\)), and let \(\nm{m} = \{ \alpha \setst \alpha \in \names \text{ or } \out{\alpha} \in \out{\names} \text{ occurs in }m\}\) be the set of (co-)names occurring in \(m\).
\end{definition}

Note that the nullary case of the sum\footnote{%
	This version of sum is used for simplifying the presentation of the LTS in \autoref{fig:ltsrules}.}%
gives the inactive process, denoted \(0\), and that the unary case gives the \enquote{usual} prefixing of a process by a label, and we write \eg \(a.P\) for \({\textstyle\sum}_{1} \lambda_i.P_i\) with \(\lambda_1 = a\) and \(P_1 = P\).
We assume sum to be associative and often consider only its binary case, that we denote with the \(+\) sign. We often forget about the trailing \(\emptyset\) in the memory stack, or the inactive process \(0\) and write \eg \(e \rhd a \mid (b + \out{c})\) for \(e.\emptymem \rhd (a.0 \mid (b.0 + \out{c}.0))\).
We work up to the structural congruence \(\congru\) of CCS---that we suppose familiar to the reader---and write \eg \(\emptymem \rhd (P_1 \mid P_2 \mid P_3)\) without parenthesis since \( (P_1 \mid P_2) \mid P_3 \congru P_1 \mid (P_2 \mid P_3)\).
Finally, alpha-equivalence is written \(=_{\alpha}\) and supposed familiar as well.

In a memory event \(\mem{i, \lambda, P}\), the \(P\) component represents the process that was not chosen in a non-deterministic transition, but that can be restored if the process wants to go back.
The \enquote{fork} symbol \(\fork\) tracks when a memory stack is split between two threads.%

\begin{definition}[Structural equivalence~{\cite%
	{Aubert2016jlamp,Danos2004}}]
	\label{def:congru_rccs}
	Structural equivalence on \(\rproc\) %
	is the smallest equivalence relation generated by the following rules:
	\begin{align}
		                  & \hspace{-2.8em} \frac{P =_{\alpha} Q}{m \rhd P \congru m \rhd Q} \tag{\(\alpha\)-Conversion} \label{eq:rcss-congru-alpha} \\
		m \rhd (P \mid Q) & \congru (\fork . m \rhd P \mid \fork . m \rhd Q)\tag{Distribution of Memory}\label{eq:rcss-congru-distrib-mem}            \\
		m \rhd P\bs a     & \congru (m \rhd P)\bs a\text{ with }a\notin \nm{m} \tag{Scope of Restriction} \label{eq:rcss-congru-scope}
	\end{align}
\end{definition}

\begin{figure}
	{
		\centering

		\begin{prooftree}
			\hypo{}
			\infer[left label=\(i \notin \ids(m)\)]1[act.]{(m \vartriangleright \lambda . P + Q) \fwlts{i}{\lambda} \mem{i, \lambda, Q} . m \vartriangleright P}
		\end{prooftree}
		\hfill
		\begin{prooftree}
			\hypo{R \fwlts{i}{\lambda} R'}
			\hypo{S \fwlts{i}{\out{\lambda}} S'}
			\infer2[syn.]{R \mid S \fwlts{i}{\tau} R' \mid S'}
		\end{prooftree}
		\\[1.8em]
		\begin{prooftree}
			\hypo{}
			\infer[left label=\(i \notin \ids(m)\)]1[act.\(_*\)]{\mem{i, \lambda, Q} . m \vartriangleright P \bwlts{i}{\lambda} m \vartriangleright (\lambda . P + Q)}
		\end{prooftree}
		\hfill
		\begin{prooftree}
			\hypo{R \bwlts{i}{\lambda} R'}
			\hypo{S \bwlts{i}{\out{\lambda}} S'}
			\infer2[syn.\(_*\)]{R \mid S \bwlts{i}{\tau} R' \mid S'}
		\end{prooftree}
		\\[1.8em]
		\begin{prooftree}
			\hypo{R \fbwlts{i}{\alpha} R'}
			\infer[left label=\(i \notin \ids(S)\)]1[par.\(_{\text{L}}\)]{R \mid S \fbwlts{i}{\alpha} R' \mid S}
		\end{prooftree}
		\hfill
		\begin{prooftree}
			\hypo{S \fbwlts{i}{\alpha} S'}
			\infer[left label=\(i \notin \ids(S)\)]1[par.\(_{\text{R}}\)]{R \mid S \fbwlts{i}{\alpha} R \mid S'}
		\end{prooftree}
		\\[1.8em]
		\hspace{3.5em}
		\begin{prooftree}
			\hypo{R \fbwlts{i}{\alpha} R'}
			\hypo{a \notin \alpha}
			\infer2[res.]{ R\bs a \fbwlts{i}{\alpha} R'\bs a}
		\end{prooftree}
		\hfill
		\begin{prooftree}
			\hypo{R_1 \congru R}
			\hypo{R\fbwlts{i}{\alpha} R'}
			\hypo{R'\congru R_1'}
			\infer3[\(\congru\)]{R_1 \fbwlts{i}{\alpha} R_1'}
		\end{prooftree}
	}
	\caption{Rules of the labeled transition system (LTS)}
	\label{fig:ltsrules}
\end{figure}

The labeled transition system for RCCS is given by the rules of \autoref{fig:ltsrules}.
We use \(\fbwlts{i}{\alpha}\) for the union of \(\redl{i:\alpha}\) (forward) and \(\revredl{i:\alpha}\) (backward transition), and if there are indices \(i_1, \hdots, i_n\) and labels \(\alpha_1, \hdots, \alpha_n\) such that \(R_1 \fbwlts{i_1}{\alpha_1} \cdots \fbwlts{i_n}{\alpha_n} R_n\), then we write \(R_1 \tfbwlts R_n\).
\autoref{sec:examples} will provide examples of executions, but it should be noted that a process \(m \rhd a. P\) is allowed to make a transition with label \(a\) and identifier \(i \notin \ids(m)\) using act.\ and add the event \(\mem{i, a, 0}\) to the memory stack \(m\) if \(i \notin \ids(m)\). %
Conversely, a process \(\mem{i, a, 0} . m \vartriangleright P\) can do a backward transition using act.\(_*\) with label \(a\) and identifier \(i\) and become \(m \rhd a. P\).
This system is a conservative extension over the LTS of CCS with prefixed sum, simply adding indices and backward transitions.

\begin{definition}[Reachable {\cite[Lemma 1]{Aubert2016jlamp}}]
	\label{def:coh}
	For all \(R\), if there is a CCS process \(P\) such that \(\emptymem\rhd P\tfbwlts R\), we say that \(R\) is \emph{reachable}, that \(P\) is the \emph{unique origin} of \(R\) and write it \(\orig{R}\).
\end{definition}

An important result~\cite[Lemma 10]{Danos2004} furthermore states that the trace \(\emptymem\rhd P\tfbwlts R\) is
forward-only\footnote{Traces and trace equivalences for RCCS are reminded in \autoref{sec:trace}: they are needed for some proofs and they are similar to their CCS's counterpart~\cite{Boudol1988}, but are not required to understand our results.}.
Also, note that multiple RCCS processes can have the same origin, but that reachable RCCS processes have one unique origin (up to structural equivalence).
We consider only reachable terms: unreachable terms are \enquote{dysfunctional} and their memory is considered \emph{not coherent}~\cite{Danos2005}, as they can not \enquote{rewind} back to an origin process.

\subsection{Processes and Memories as (Identified) Configuration Structures}
\label{sec:css-enc}
\label{sec:enc_mem}

In the definitions below, we write \(S\) for a synchronization algebra \((S, \star,\bot)\) with \(S = \names \cup \out{\names}\cup\{\tau\}\).

\begin{definition}[{Encoding CCS processes~\cite%
				{Winskel1995}}]
	\label{def:encoding-CCS-cf}
	Given a CCS process \(P\), its encoding \(\enc{P}\) as a configuration structure is built inductively:
	\begin{align*}
		\enc{\lambda .P + Q} & =\enc{\lambda .P}+\enc{Q}
		                     &                                 &  &
		\enc{P\mid Q}        & =\enc{P}\mid_{S}\enc{Q}         &  &   &
		\enc{P \bs a}        & =\enc{P}\crestr{\{a, \out{a}\}}          \\
		\enc{\lambda . P}    & =\lambda .\enc{P}               &  &   &
		\enc{0}              & = \confzero
	\end{align*}
	where \(S\) includes
	\(\alpha \bullet \out{\alpha} = \tau\) and \(\alpha \bullet \beta = \bot\), if \(\beta\neq \out{\alpha}\).
\end{definition}

\begin{definition}[Encoding RCCS memories]
	\label{def:encode_m}
	Given a RCCS process \(R\), the encoding \(\encmp{R}\) of its memory as an \(\iconf\)-structure is built by induction on the process and on the memory:
	\begin{align*}
		\encmp{m \rhd P}               & = \encm{m}                        &  &  &
		\encmp{R_1 \mid R_2}           & = \encmp{R_1} \mid_{S}\encmp{R_2} &  &  &
		\encmp{R\bs a}                 & = \encmp{R}                               \\
		\encm{\mem{i, \lambda, P} . m} & = \encm{m} \post (\lambda, i)     &  &  &
		\encm{\emptymem}               & = \iconfzero                      &  &  &
		\encm{\fork . m}               & = \encm{m}
	\end{align*}
	where \(S\) includes
	\(\alpha \bullet \out{\alpha} = \tau\); \(\alpha \bullet \alpha = \alpha\) and \(\alpha \bullet \beta = \bot\) if \(\beta\notin\{\out{\alpha}, \alpha\}\).
\end{definition}

\subsection{Examples}
\label{sec:examples}

We now illustrate the execution of RCCS processes, the encoding of CCS processes and of RCCS memories, and how they relate.

\begin{example}[Executing a RCCS process]
	\label{ex:proc-lts-1}
	An example of forward-only trace is:
	\begin{alignat}{4}
		\emptymem \rhd (a.b \mid c.\out{a}) & \congru         & (\fork . \emptymem \rhd a.b)                               & \mid (\fork . \emptymem \rhd c.\out{a}) \tag{\ref{eq:rcss-congru-distrib-mem}} \\
		                                    & \fwlts{1}{c}    & (\fork . \emptymem \rhd a.b)                               & \mid(\mem{1,c,0}. \fork . \emptymem \rhd \out{a}) \tag{act.}                   \\
		                                    & \fwlts{2}{\tau} & (\mem{2, a, 0} . \fork . \emptymem \rhd b)                 & \mid (\mem{2, \out{a}, 0} . \mem{1,c,0} . \fork .\emptymem \rhd 0)\tag{syn.}   \\
		                                    & \fwlts{3}{b}    & (\mem{3, b, 0} . \mem{2, a, 0} . \fork . \emptymem \rhd 0) & \mid (\mem{2, \out{a}, 0} . \mem{1,c,0} . \fork . \emptymem \rhd 0) \tag{act.}
	\end{alignat}
	Reading it from end to beginning and replacing \(\fwlts{\_}{\_}\) with \reflectbox{\(\bwlts{\_}{\_}\)} gives a \emph{backward-only} trace, that would rewind the process back to its origin.
	Of course, a trace can mix forward and backward transitions, as we illustrate in \autoref{ex:op-cor}.
	The memory of this process is encoded in \autoref{ex:proc}.
\end{example}

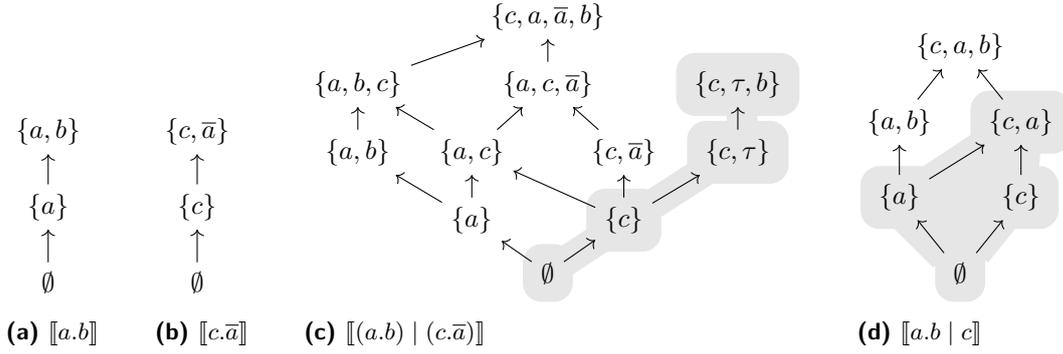
\begin{figure}
	\begin{minipage}[b]{.1\linewidth}
		\begin{tikzpicture}%
			\node (emptyset) at (0,0) {\(\emptyset\)};
			\node (a1) at (0, 1) {\(\{a\}\)};
			\node (a2) at (0, 2) {\(\{a, b\}\)};
			\draw [->] (emptyset) -- (a1);
			\draw [->] (a1) -- (a2);
		\end{tikzpicture}
		\subcaption{\(\enc{a.b}\)}
		\label{fig:ex_mem1}
	\end{minipage}
	\hfill
	\begin{minipage}[b]{.1\linewidth}
		\begin{tikzpicture}%
			\node (emptyset) at (0,0) {\(\emptyset\)};
			\node (a1) at (0, 1) {\(\{c\}\)};
			\node (a2) at (0, 2) {\(\{c, \out{a}\}\)};
			\draw [->] (emptyset) -- (a1);
			\draw [->] (a1) -- (a2);
		\end{tikzpicture}
		\subcaption{\(\enc{c.\out{a}}\)}
		\label{fig:ex_mem2}
	\end{minipage}
	\hfill
	\begin{minipage}[b]{.48\linewidth}
		\begin{tikzpicture}%
			\node (emptyset) at (0,0) {\(\emptyset\)};
			\node (a1) at (-1, 0.7) {\(\{a\}\)};
			\node (a2) at (1, 0.7) {\(\{c\}\)};
			\node (b1) at (-2.5, 1.6) {\(\{a, b\}\)};
			\node (b2) at (-1, 1.6) {\(\{a, c\}\)};
			\node (b3) at (1, 1.6) {\(\{c, \out{a}\}\)};
			\node (b4) at (2.5, 1.6) {\(\{c, \tau\}\)};
			\node (c1) at (-2.5, 2.5) {\(\{a, b, c\}\)};
			\node (c2) at (0, 2.5) {\(\{a, c, \out{a}\}\)};
			\node (c3) at (2.5, 2.5) {\(\{c, \tau, b\}\)};
			\node (d1) at (0, 3.4) {\(\{c, a, \out{a}, b\}\)};
			\draw [->] (emptyset) -- (a1);
			\draw [->] (emptyset) -- (a2);
			\draw [->] (a1) -- (b1);
			\draw [->] (a1) -- (b2);
			\draw [->] (a2) -- (b2);
			\draw [->] (a2) -- (b3);
			\draw [->] (a2) -- (b4);
			\draw [->] (b1) -- (c1);
			\draw [->] (b2) -- (c1);
			\draw [->] (b2) -- (c2);
			\draw [->] (b3) -- (c2);
			\draw [->] (b4) -- (c3);
			\draw [->] (c1) -- (d1);
			\draw [->] (c2) -- (d1);
			\begin{scope}[on background layer]
				\node[fill=gray!20, fit=(emptyset), rounded corners=.25cm] {};
				\node[fill=gray!20, fit=(a2), rounded corners=.25cm] {};
				\node[fill=gray!20, fit=(b4), rounded corners=.25cm] {};
				\node[fill=gray!20, fit=(c3), rounded corners=.25cm] {};
				\draw[gray!20, line width=4mm%
					, rounded corners
				] (emptyset) -- ($(a2)+(-0.1, -0.1)$) -- ($(b4)+(0.1, 0.1)$) -- (c3);
			\end{scope}
		\end{tikzpicture}
		\subcaption{\(\enc{(a.b)\mid (c.\out{a})}\)}\label{fig:ex_mem3}

	\end{minipage}
	\hfill
	\begin{minipage}[b]{.2\linewidth}
		\begin{tikzpicture}%
			\node (emptyset) at (0,0) {\(\emptyset\)};
			\node (a1) at (-0.8, 1) {\(\{a\}\)};
			\node (a2) at (-0.8, 2) {\(\{a, b\}\)};
			\node (a4) at (0.8, 2) {\(\{c, a\}\)};
			\node (a3) at (0.8, 1) {\(\{c\}\)};
			\node (a5) at (0, 3) {\(\{c, a, b\}\)};
			\draw [->] (emptyset) -- (a1);
			\draw [->] (emptyset) -- (a3);
			\draw [->] (a1) -- (a2);
			\draw [->] (a1) -- (a4);
			\draw [->] (a3) -- (a4);
			\draw [->] (a4) -- (a5);
			\draw [->] (a2) -- (a5);
			\begin{scope}[on background layer]
				\node[fill=gray!20, fit=(emptyset), rounded corners=.25cm] {};
				\node[fill=gray!20, fit=(a1), rounded corners=.25cm] {};
				\node[fill=gray!20, fit=(a4), rounded corners=.25cm] {};
				\node[fill=gray!20, fit=(a3), rounded corners=.25cm] {};
				\draw[gray!20, line width=4mm%
					, rounded corners
				] (emptyset) -- ($(a1)+(-0.1, -0.1)$) -- ($(a4)+(0.1, 0.1)$) -- (a3) -- (emptyset);
				\path[fill=gray!20, line width=4mm, rounded corners] (emptyset.center) to (a1.center) to (a4.center) to (a3.center) to (emptyset.center);
			\end{scope}
		\end{tikzpicture}
		\subcaption{\(\enc{a.b \mid c}\)}\label{fig:ex_mem4}
	\end{minipage}

	\caption{\(\iconf\)-structures for Examples \ref{ex:proc-enc-1}, \ref{ex:proc}, \ref{ex:proc2} and
		\ref{ex:op-cor}, with the CCS term their underlying configuration encode in caption.}
	\label{fig:cs1}
\end{figure}

\begin{example}[Encoding CCS processes]
	\label{ex:proc-enc-1}
	We can see the \(\iconf\)-structures from \autoref{fig:ex} as configuration structures obtained by encoding the CCS processes \(a + a\), \(a\mid a\), and \((a.b)\mid \out{a}\).
	Similarly, we can consider the \(\iconf\)-structures from \autoref{fig:cs1}---ignoring the grayed out parts for now-- as configuration structures. The interested reader can check that the encoding of \((a.b) \mid (c.\out{a})\) in \autoref{fig:ex_mem3} is indeed the result of applying the parallel composition of configurations structures (\autoref{def:par-comp-ccs}) to the encoding of \(a.b\) in \autoref{fig:ex_mem1} and \(c.\out{a}\) in \autoref{fig:ex_mem2}. Lastly, \autoref{fig:ex_mem4} shows the encoding of \((a.b) \mid c\).

\end{example}

The parallel composition of \(\iconf\)-structures (\autoref{def:par-ics}) differs slightly and is new, and hence deserves a detailed example.

\begin{example}[Encoding RCCS memories]
	\label{ex:proc}
	The process obtained at the end of \autoref{ex:proc-lts-1} has its memory encoded as follows:
	\begin{align*}
		  & \encmp{(\mem{3, b, 0} . \mem{2, a, 0} . \fork . \emptymem \rhd 0) \mid (\mem{2, \out{a}, 0} . \mem{1,c,0} . \fork . \emptymem \rhd 0)}                           \\
		= & \encmp{\mem{3, b, 0} . \mem{2, a, 0} . \fork . \emptymem \rhd 0} \mid \encmp{\mem{2, \out{a}, 0} . \mem{1,c,0} . \fork . \emptymem \rhd 0}                       \\
		= & \encm{\mem{3, b, 0} . \mem{2, a, 0} . \fork . \emptymem} \mid \encm{\mem{2, \out{a}, 0} . \mem{1,c,0} . \fork . \emptymem}
		\shortintertext{Letting \(E = L = \{a, \out{a}, b, c\}\),
			\(\labl = \id\)%
			, \(I = \{1, 2, 3\}\), using \(\encm{\fork . \emptymem} = \encm{\emptymem} = \iconfzero\) and the postfixing:}
		  & \encm{\mem{3, b, 0}.\mem{2, a, 0}.\fork.\emptymem} = (\{a, b\}, \{ \emptyset, \{a\}, \{a, b\}\}, L, \labl, I, \{a \mapsto 2, b \mapsto 3
		\}                                                                                                                                                                   \\
		  & \encm{\mem{2, \out{a}, 0}.\mem{1, c, 0}.\fork.\emptymem} = (\{c, \out{a}\}, \{ \emptyset, \{c\}, \{c,\out{a}\} \}, L, \labl, I, \{c \mapsto 1, \out{a} \mapsto 2
		\})
	\end{align*}

	Those are displayed in Fig.~\ref{fig:ex_mem1} and~\ref{fig:ex_mem2}, and their product (which is the first step to compute their parallel composition) gives the following sets of events and identifiers:

	\begin{center}
		\arrayrulecolor{gray!20}
		\(	\begin{array}{|c|c|c|c|c|c|c|c|c|c|}
			\hline
			\text{Event}                         & (a, \star) & (b, \star) & (\star, c) & (\star, \out{a}) & (a, c) & (a, \out{a}) & (b, c) & (b, \out{a}) \\
			\rowcolor{gray!20} \text{Identifier} & (2, \star) & (3, \star) & (\star, 1) & (\star, 2)       & (2, 1) & (2, 2)       & (3, 1) & (3, 2)       \\
			\hline
		\end{array}\)
	\end{center}

	Re-identifying and re-labeling according to the definition of parallel composition gives:
	\begin{center}
		\(
		\arrayrulecolor{gray!20}
		\begin{array}{|c|c|c|c|c|c|c|c|c|c|}
			\hline
			\text{Event}                            & (a, \star)        & (b, \star) & (\star, c) & (\star, \out{a})  & (a, c)       & (a, \out{a}) & (b, c)       & (b, \out{a}) \\
			\rowcolor{gray!20} \text{Re-identified} & \bot_{(2, \star)} & 3          & 1          & \bot_{(\star, 2)} & \bot_{(2,1)} & 2            & \bot_{(3,1)} & \bot_{(3,2)} \\
			\text{Re-labeled}                       & \bot              & b          & c          & \bot              & \bot         & \tau         & \bot         & \bot         \\
			\hline
		\end{array}
		\)
	\end{center}

	Indeed, if two events occur at the same time with the same identifier (\ref{valid-synch}), then their identifier is simply picked.
	Hence, \(\iden'(a, \out{a}) = 2\).
	If only one event is present in the pair, and no event on the other component have the same identifier (\ref{extra1}, \ref{extra2}), then this event's identifier is picked.
	Hence, \(\iden'(b, \star) = 3\) and \(\iden'(\star, c) = 1\).
	The remaining cases %
	get re-identified with \(\bot_k\) (\ref{error}). %
	Finally, %
	\((b, \star)\), \((\star, c)\) and \((a, \out{a})\) gets relabeled with \(b\), \(c\) and \(\tau\) respectively, %
	and after restricting to the label \(\bot\) we obtain the
	grayed out part of \autoref{fig:ex_mem3}.%

\end{example}

Observe that in this last example, the structure underlying the encoding of the memory is just a particular \enquote{path} in the encoding of the origin.
We can observe this intuition again with the following example: %

\begin{example}[Memory and origin]
	\label{ex:proc2}
	The encoding of the memory resulting from the execution
	\(\emptymem\rhd ((a.b)\mid c) \fwlts{1}{c} \fwlts{2}{a} ( \mem{2, a, 0} . \fork . \emptymem \rhd b) \mid (\mem{1,c,0} . \fork . \emptymem \rhd 0)\)
	is the grayed out part in \autoref{fig:ex_mem4}, with \(\iden(c) = 1\) and \(\iden(a) = 2\). %
	We name this process \(R_1\) and come back to it in \autoref{ex:op-cor}.
\end{example}

\subsection{Operational Correspondence}
\label{sec:op-corr}

Before studying bisimulations on configuration structures and processes, we prove the operational correspondence\footnote{Similar to the operational correspondance between configuration structures and CCS processes~\cite[Section 3]{Boudol1987} \ie, if \(P \redl{\alpha} Q\), then \(\enc{P} = \enc{Q}\crestr{\{e\}}\), where \(e\) is an event in \(\enc{P}\) such that \(\labl(e)=\alpha\).} between RCCS processes and the encodings of their memories in \(\iconf\)-structures (\autoref{lem:transition_identified}, \cf also \autoref{sub:prop-mem-proof}).
Events in \(\iconf\)-structures resulting from the encoding of a process have different identifiers, and they are either causally linked or concurrent.%

\begin{restatable}[Memories give posets]{lemma}{prefixenc}
	\label{lem:prefix_enc}
	For all \(R\), letting \(x\) be the maximal configuration in \(\encmp{R}\) (\autoref{def:causality}), \((\encmp{R}, \subseteq)\) is a partially ordered set (poset) with maximal element \(x\).
\end{restatable}

This is proved by induction on \(R\) and illustrated by Examples~\ref{ex:proc} and \ref{ex:proc2}. %
However, having at most one maximal configuration does not imply that one particular event has to be \enquote{the last one} introduced.
We use the following definition to make it formal.

\begin{definition}[Maximal event]
	\label{def:maximal-event}
	An event \(e\) is \emph{maximal in \(\iconf\)} %
	if there is no event \(e'\) such that \(e<_{x}e'\), for \(x\) a maximal configuration of \(\iconf\).
\end{definition}

For instance, %
the encoding of the memory of \autoref{ex:proc2}, pictured in \autoref{fig:ex_mem4}, has two maximal events, labeled \(a\) and \(c\).
We can now state the main result of this section:

\begin{restatable}[Operational Correspondence]{lemma}{transitionidentified}
	\label{lem:transition_identified}
	For all \(R\) and \(S\), writing \((E_R, C_R, \labl_R, I_R, \iden_R)\) for \(\encmp{R}\) and similarly for \(S\), if \(R\fwlts{i}{\alpha}S\) or \(S\bwlts{i}{\alpha}R\), then
	there exists \(e \in E_S\) maximal in \(\encmp{S}\) with \(\iden_S(e) = i\) \st \(\encmp{R} \iso \encmp{S}\crestr{\{e\}}\). %
	For all \(R\) and \(e\) a maximal event in \(\encmp{R}\), there is a transition \(R\text{\raisebox{-.10pt}{\(\bwlts{\iden_R(e)}{\labl_R(e)}\)}}S\) with \(\encmp{S} \iso \encmp{R}\crestr{\{e\}}\).
\end{restatable}

For the first part, it suffices to
show that the forward transition triggers the creation of a maximal event with the same identifier, and nothing else, and that this event can be \enquote{traced} in \(\encmp{S}\).
It uses intermediate lemma (Lemmas~\ref{lem:equivalence_encoding}--\ref{lem:prevent-max-event}) showing how maximal events are preserved by certain operations on \(\iconf\)-structures.
The result follows easily for backward transition, but %
the last part is more involved: it requires to show that \(e\) can be mapped to a particular transition in the trace from \(\orig{R}\) to \(R\), and, using a notion of trace equivalence, that this particular transition can be \enquote{postponed} and done last, so that \(R\) can backtrack on it.%

\begin{example}[Forward and backward transitions]
	\label{ex:op-cor}
	Looking back at the process of \autoref{ex:proc2}, we could further have \(R_1 \bwlts{1}{c} R_2 \fwlts{3}{b} R_3\), \ie
	\begin{align*}
		( \mem{2, a, 0} . \fork . \emptymem \rhd b) \mid (\mem{1,c,0} . \fork . \emptymem \rhd 0) & \bwlts{1}{c} ( \mem{2, a, 0} . \fork . \emptymem \rhd b) \mid \fork . \emptymem \rhd c)) \tag{act.\(_*\)}           \\
		                                                                                          & \fwlts{3}{b} ( \mem{3, b, 0} . \mem{2, a, 0} . \fork . \emptymem \rhd 0) \mid \fork . \emptymem \rhd c)) \tag{act.}
	\end{align*}
	We can see using \autoref{fig:ex_mem4} that \(\encmp{R_1}\crestr{c} = \encmp{R_2}\) and that \(\encmp{R_2} = \encmp{R_3} \crestr{b}\).
\end{example}

\section{Reversible and Truly Concurrent Bisimulations Are the Same}

\subsection{History-Preserving Bisimulations in Configuration Structures}
\label{sec:hpbs}
History-preserving bisimulation (\HPB)~%
\cite{Baldan2014, Phillips2014,Rabinovich1988}
and hereditary history-preserving bisimulation (\HHPB)~%
\cite{Baldan2014, Bednarczyk1991} are equivalences on configuration structures that use label- and order-preserving bijections on events.
Below, assume given \(\conf_i = (E_i, C_i, L_i, \labl_i)\) for \(i = 1, 2\).

\begin{definition}[Label- and order-preserving functions (\lop)]
	\label{def:label-order-functions}
	A function \(f: x_1 \to x_2\), for \(x_i\in C_i\), \(i\in\{1,2\}\) is \emph{label-preserving} if \(\labl_1(e) = \labl_2(f(e))\) for all \(e \in x_1\).
	It is \emph{order-preserving} if \(e_1 \leqslant_{x_1} e_2 \Rightarrow f(e_1) \leqslant_{x_2} f(e_2)\), for all \(e_1, e_2 \in x_1\).
	We write that \(f\) is \emph{\lop} if it is both.
\end{definition}

\begin{definition}[\HPB and \HHPB]
	\label{def:hpb}
	A relation \(\rel\subseteq C_1\times C_2\times (E_1 \rightharpoonup E_2)\) such that \((\emptyset,\emptyset,\emptyset)\in {\rel}\), and if \((x_1,x_2,f)\in {\rel}\), then \(f\) is a \lop bijection between \(x_1\) and \(x_2\) and (\ref{hpb1}) and (\ref{hpb2}) (\resp (\ref{hpb1}--\ref{hpb4})) hold is called a \emph{history-} (\resp \emph{hereditary history-}) \emph{preserving bisimulation (\HPB, \resp \HHPB) between \(\conf_1\) and \(\conf_2\)}.
	\begin{align}
		\forall y_1, x_1\redl{e_1}y_1 \Rightarrow
		\exists y_2, g, x_2\redl{e_2}y_2, g\frestr{x_1}
		= f, (y_1, y_2, g)\in {\rel} \label{hpb1} \\
		\forall y_2, x_2\redl{e_2}y_2 \Rightarrow
		\exists y_1, g, x_1\redl{e_1}y_1, g\frestr{x_1}
		= f, (y_1, y_2, g)\in {\rel} \label{hpb2} \\
		\forall y_1, x_1 \revredl{e_1} y_1 \Rightarrow
		\exists y_2, g, x_2 \revredl{e_2} y_2, g = f\frestr{y_1}%
		, (y_1, y_2, g)\in {\rel}\label{hpb3}     \\
		\forall y_2, x_2 \revredl{e_2} y_2 \Rightarrow
		\exists y_1, g, x_1 \revredl{e_1} y_1, g = f\frestr{y_1} %
		, (y_1, y_2, g)\in {\rel} \label{hpb4}
	\end{align}

	We write that \emph{\(\conf_1\) and \(\conf_2\) are \HorHPB} if there exists a \HorHPB relation between them.
\end{definition}

Note that \HPB and \HHPB are two different relations, as \eg \((a\mid (b+c)) + (a\mid b) + ((a+c) \mid b)\) and \((a\mid(b+c)) + ((a+c)\mid b)\) have \HPB but not \HHPB encodings~\cite{Glabbeek2001} .

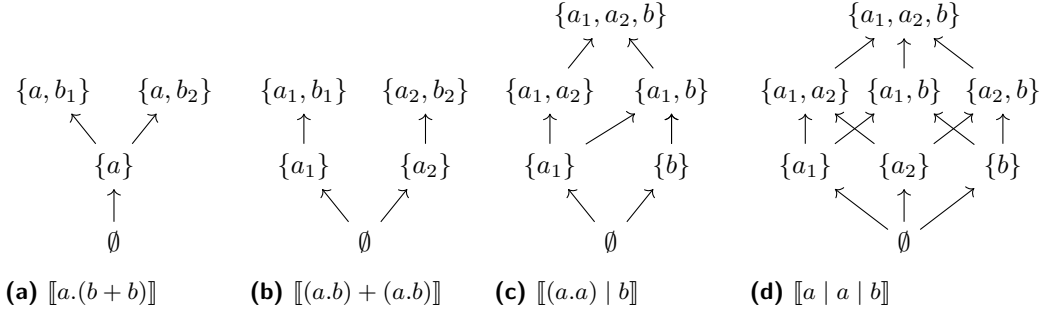
\begin{figure}
	\begin{minipage}[b]{.22\linewidth}
		\begin{tikzpicture}%
			\node (emptyset) at (1, -1) {\(\emptyset\)};
			\node (a) at (1,0) {\(\{a\}\)};
			\node (b1) at (0.2, 1) {\(\{a, b_1\}\)};
			\node (b2) at (1.8, 1) {\(\{a, b_2\}\)};
			\draw [->] (emptyset) -- (a);
			\draw [->] (a) -- (b1);
			\draw [->] (a) -- (b2);
		\end{tikzpicture}
		\subcaption{\(\enc{a. (b + b)}\)}\label{fig:ex:notH3}
	\end{minipage}
	\hfill
	\begin{minipage}[b]{.22\linewidth}
		\begin{tikzpicture}%
			\node (emptyset) at (1, -1) {\(\emptyset\)};
			\node (a1) at (0.2,0) {\(\{a_1\}\)};
			\node (a2) at (1.8,0) {\(\{a_2\}\)};
			\node (b1) at (0.2, 1) {\(\{a_1, b_1\}\)};
			\node (b2) at (1.8, 1) {\(\{a_2, b_2\}\)};
			\draw [->] (emptyset) -- (a1);
			\draw [->] (emptyset) -- (a2);
			\draw [->] (a1) -- (b1);
			\draw [->] (a2) -- (b2);
		\end{tikzpicture}
		\subcaption{\(\enc{(a. b) + (a.b)}\)}\label{fig:ex:notH4}
	\end{minipage}
	\hfill
	\begin{minipage}[b]{.23\linewidth}
		\begin{tikzpicture}%
			\node (emptyset) at (0,0) {\(\emptyset\)};
			\node (a) at (-0.8,1) {\(\{a_1\}\)};
			\node (b) at (0.8, 1) {\(\{b\}\)};
			\node (ab) at (0.8, 2) {\(\{a_1, b\}\)};
			\node (aa) at (-0.8, 2) {\(\{a_1, a_2\}\)};
			\node[align=center] (aab) at (0, 3) {\(\{a_1, a_2, b\}\)};
			\draw [->] (emptyset) -- (a);
			\draw [->] (emptyset) -- (b);
			\draw [->] (a) -- (aa);
			\draw [->] (a) -- (ab);
			\draw [->] (b) -- (ab);
			\draw [->] (b) -- (ab);
			\draw [->] (ab) -- (aab);
			\draw [->] (aa) -- (aab);
		\end{tikzpicture}
		\subcaption{\(\enc{(a . a) \mid b}\)}\label{fig:ex:notH1}
	\end{minipage}
	\hfill
	\begin{minipage}[b]{.3\linewidth}
		\begin{tikzpicture}%
			\node (emptyset) at (1, -1) {\(\emptyset\)};
			\node (a1) at (-0.3,0) {\(\{a_1\}\)};
			\node (a2) at (1,0) {\(\{a_2\}\)};
			\node (b) at (2.3, 0) {\(\{b\}\)};
			\node (aa) at (-0.3, 1) {\(\{a_1,a_2\}\)};
			\node (a1b) at (1,1) {\(\{a_1, b\}\)};
			\node (a2b) at (2.3, 1) {\(\{a_2, b\}\)};
			\node[align=center] (aab) at (1,2) {\(\{a_1, a_2, b\}\)};
			\draw [->] (emptyset) -- (a1);
			\draw [->] (emptyset) -- (a2);
			\draw [->] (emptyset) -- (b);
			\draw [->] (a1) -- (a1b);
			\draw [->] (a2) -- (a2b);
			\draw [->] (a1) -- (aa);
			\draw [->] (a2) -- (aa);
			\draw [->] (b) -- (a1b);
			\draw [->] (b) -- (a2b);
			\draw [->] (a1b) -- (aab);
			\draw [->] (a2b) -- (aab);
			\draw [->] (aa) -- (aab);
		\end{tikzpicture}
		\subcaption{\(\enc{a \mid a \mid b}\)}\label{fig:ex:notH2}
	\end{minipage}
	\caption{Configuration structures for Examples~\ref{ex:HPB-not-HHPB} and \ref{ex:not_hhpb}}
	\label{fig:ex2}
\end{figure}

\begin{example}
	\label{ex:HPB-not-HHPB}
	The encoding of the two processes \(a.(b+b)\) and \((a.b)+(a.b)\), in Fig.~\ref{fig:ex:notH3} and \ref{fig:ex:notH4}, are \HHPB: the relation can start
	by mapping \(a\) to \(a_1\) or \(a_2\),
	and then maps \(b_1\) or \(b_2\) (depending on the superset reached) to \(b_1\) or \(b_2\), according to the first choice made.
	This relation can \enquote{follow} the forward and backward movements in both structures, giving a \lop bijection.
	This example also proves that \HHPB is not CCS's structural congruence.
\end{example}

\subsection{Back-and-forth Bisimulations in Reversible CCS}
\label{sec:HHPB-CCS}

This section presents the relations we will be using, explain the restrictions on previous attempts to capture \HHPB as a relation on process algebra terms~\cite{Aubert2016jlamp,Phillips2007}, and shows why both backward transitions \emph{and} identifiers are needed to capture \HHPB.

Below, assume given two reachable processes \(R_1\) and \(R_2\), and if \(f : A \to B\) is such that \(f(a) = b\), we write \(f \setminus \{a \mapsto b\}\) for \(f \frestr{A{\setminus}\{a\}}\).

\begin{definition}[\textnormal{B\&F} and \textnormal{SB\&F} bisimulations]
	\label{def:hhpb_ident}
	A relation \(\rel\subseteq \rproc \times \rproc \times (\ids \rightharpoonup \ids)\) such that \((\emptymem\rhd \orig{R_1}, \emptymem\rhd \orig{R_2}, \emptymem)\in\rel\) and if \((R_1, R_2, f) \in \rel\), then \(f\) is a bijection between \(\ids(R_1)\) %
	and \(\ids(R_2)\)
	and (\ref{bf1}--\ref{bf4}) hold is called \emph{a back-and-forth-bisimulation (\BF) between \(R_1\) and \(R_2\)}.
	\begin{align}
		\forall S_1, R_1\fwlts{i}{\alpha}S_1 \Rightarrow
		\exists S_2, g, R_2\fwlts{j}{\alpha}S_2, g = f\cup\{i \mapsto j\}, (S_1, S_2, g)\in\rel \label{bf1}        \\
		\forall S_2, R_2\fwlts{i}{\alpha}S_2 \Rightarrow
		\exists S_1, g, R_1\fwlts{j}{\alpha}S_1, g = f\cup\{i \mapsto j\}, (S_1, S_2, g)\in\rel \label{bf2}        \\
		\forall S_1, R_1\bwlts{i}{\alpha}S_1 \Rightarrow
		\exists S_2, f, R_2\bwlts{j}{\alpha}S_2, g = f{\setminus}\{i \mapsto j\}, (S_1, S_2, g)\in\rel \label{bf3} \\
		\forall S_2, R_2\bwlts{i}{\alpha}S_2 \Rightarrow
		\exists S_1, g, R_1\bwlts{j}{\alpha}S_1, g = f{\setminus}\{i \mapsto j\}, (S_1, S_2, g)\in\rel \label{bf4}
	\end{align}
	If we remove the requirements on \(f\) and \(g\) in the second part of (\ref{bf1}--\ref{bf4}), we call \(\rel\) a \emph{simple back-and-forth bisimulation (\SBF)}.
	We write that \emph{\(R_1\) and \(R_2\) are \BF} (\resp \SBF) if there exists a \BF (\resp \SBF) relation between them.
\end{definition}

\subparagraph*{Restrictions and Previous Results}

\begin{definition}[Constraints]%
	\label{def:autoc}
	\label{def:non-repeat-cs}
	\label{def:singly-labeled}
	Given \(\conf\), if \(\forall x \in C\),
	\(\forall e_1, e_2 \in x\), \(\labl(e_1) = \labl(e_2)\) implies \(e_1 = e_2\), then \(\conf\) is \emph{non-repeating}~\cite[Definition 3.5]{Phillips2007}.
	If, \(\forall x \in C\), \(\forall e_1, e_2 \in x\), \(e_1 \co_x e_2\) and \(\labl(e_1) = \labl(e_2)\) implies \(e_1 = e_2\), then \(\conf\) is \emph{without auto-concurrency}~\cite[Definition 9.5]{Glabbeek2001}.
	A process \(R\) is \emph{non-repeating} (\resp \emph{without auto-concurrency}) if \(\enc{\orig{R}}\) is.
\end{definition}

Those are the constraints used in showing equivalences between process algebra and configuration structures.
We omitted the definition of \emph{singly labeled}~\cite[Definition 26]{Aubert2016jlamp}, as it does not contribute to the understanding of our results.
Every non-repeating process is without auto-concurrency, but being non-repeating events and singly labeled are incomparable features. %
Simple processes can be repeating (\eg \(a.a\)), with auto-concurency (\eg \((a.b) \mid a\)), not singly labeled (\eg \(a + a\)), as well as complicated processes using these patterns. %

The first syntactical characterization of \HHPB was obtained on non-repeating processes, using the \enquote{forward-reverse bisimulation} (\FR)~\cite[Definition 6.5]{Phillips2007b}, which is essentially defined as \BF, with the additional requirement that \(f = \id\).
The theorem states that non-repeating CCSK processes are \FR iff their encoding are \HHPB~\cite[Theorem 5.4]{Phillips2007}.
We argue that \FR gives too much importance to the technical apparatus implementing reversibility: processes should be able to pick their identifiers freely, and comparing them when establishing bisimulations should not require identities, but only bijections.

A second attempt~\cite{Aubert2016jlamp} to capture %
\HHPB %
used a \emph{back-and-forth barbed congruence} on RCCS processes which was proven to correspond to \HHPB on their encodings for a restricted class of processes as well, the class of singly-labeled processes. %

\subparagraph*{Pinpointing the Right Reversible Bisimulation}
We lift both restrictions in \autoref{cor:previousresult}, %
by proving that \BF captures \HHPB on \emph{all} processes.
Before doing so, let us note that even though \BF is the right notion to capture \HHPB, when restricted to non-repeating processes, which are also without auto-concurrency, it does not use in a meaningful way %
the identifiers. %

\begin{restatable}[Collapsing \BF and \SBF]{theorem}{thequivalencebf}
	\label{th:equivalence_bf}
	If \(R_1\) and \(R_2\) are without auto-concurrency, then they are \BF iff they are \SBF.%
\end{restatable}

The intuition---made formal in \autoref{subsec:collapse-proof}---is that since two concurrent transitions sharing the same label can not be fired at the same time,
the identifiers do not add any information.
The proof is easy for the forward transitions, and uses an order on the transitions enforced by causality for the backward traces.
In the presence of auto-concurrency, the relations differ, \eg the process with auto-concurrency \(a\mid a\) and \(a.a\) are \SBF but not \BF.

\begin{example}[Reversibility is \emph{not} \enquote{just back and forth}]
	\label{ex:not_hhpb}
	Observe that the bisimulation relation obtained by only considering (\ref{bf1}--\ref{bf2}) and ignoring the identifiers in \autoref{def:hhpb_ident} is the \enquote{standard} CCS bisimulation. %
	Hence, it could seem natural to assume that \enquote{simply adding the backwards transitions}, \ie taking (\ref{bf1}--\ref{bf4}) without the identifiers, giving \SBF, would be \enquote{the right} bisimulation for RCCS.
	Processes like
	\((a.a) \mid b\) and %
	\(a \mid a \mid b\) are %
	\SBF, but %
	their encodings, presented in Fig.~\ref{fig:ex:notH1} and \ref{fig:ex:notH2}, are not \HPB and hence not \HHPB: \SBF does not account for reversibility in a satisfactory manner.

\end{example}

Both the bijection on identifiers \emph{and} backward transitions are necessary to capture \HHPB. Indeed, as suggested by \autoref{ex:not_hhpb}, \enquote{simply} considering forward and backward transitions is not enough. Let us now consider the role of the bijection on identifiers a bit further. A first remark is that \autoref{th:equivalence_bf} shows that it is easy to overlook the role of identifiers when restricting the class of processes considered. Secondly, we can prove, as an immediate corolary of Theorems \ref{th:hpb_rccs} and \ref{thm:def_hpb_coincide}, that considering only (\ref{bf1}--\ref{bf2}) (\emph{with} the identifiers) in \autoref{def:hhpb_ident} gives a characterization of \HPB (\autoref{cor:fw-with-id}): if anything, having a bijection between identifiers--thanks to the order on events that can be deducted from it--helps getting closer to \enquote{truly concurrent} bisimulation than adding backward transitions does. However, as \HPB and \HHPB do not coincide, the identifiers are not enough either.

Of course, similar mechanisms could achieve similar results, but it is our hope that reversibility is fully understood as not \enquote{just} being about adding backward transitions or memories, but to use both to obtain \enquote{backward determinism}.

\subsection{History-Preserving Bisimulations in (R)CCS}
\label{sec:lifting}

Proving our main result (\autoref{cor:previousresult})
will use intermediate relations on processes---called \HPB and \HHPB as well---that use the encoding of the memories into \(\iconf\)-structures. %
Those relations are proven to correspond to \HorHPB on the encoding of the processes on one hand (\autoref{th:hpb_rccs}), and the one %
that characterizes \HHPB is proven to coincide with \BF (\autoref{thm:def_hpb_coincide}) on the other hand.
The proofs and connections between formalisms are gathered in \autoref{subsec:proof-main}.

\begin{definition}[\HPB and \HHPB on RCCS]
	\label{def:hpb_rccs}
	A relation \(\rel\subseteq \rproc \times \rproc \times (E_1 \rightharpoonup E_2)\) such that \((\emptymem \vartriangleright \orig{R_1}, \emptymem \vartriangleright \orig{R_2}, \emptyset) \in\rel\) and if \((R_1, R_2, f) \in \rel\) then \(f\) is an isomorphism between \(\encmp{R_1}\) and \(\encmp{R_2}\) and (\ref{hpb_rccs1}) and (\ref{hpb_rccs2}) (\resp (\ref{hpb_rccs1}--\ref{hpb_rccs4})) hold is called a \emph{history-}(\resp \emph{hereditary history-}) \emph{preserving bisimulation between \(R_1\) and \(R_2\)}.
	\begin{align}
		\forall S_1, R_1\fwlts{i}{\alpha}S_1 \Rightarrow
		\exists S_2, g, R_2\fwlts{j}{\alpha}S_2, g\frestr{\encmp{R_1}} = f, (S_1, S_2, g)\in\rel \label{hpb_rccs1}  \\
		\forall S_2, R_2\fwlts{i}{\alpha}S_2 \Rightarrow
		\exists S_1, g, R_1\fwlts{j}{\alpha}S_1, g\frestr{\encmp{R_1}} = f , (S_1, S_2, g)\in\rel \label{hpb_rccs2} \\
		\forall S_1, R_1\bwlts{i}{\alpha}S_1 \Rightarrow
		\exists S_2, f, R_2\bwlts{j}{\alpha}S_2, g = f\frestr{\encmp{S_1}}, (S_1, S_2, g)\in\rel \label{hpb_rccs3}  \\
		\forall S_2, R_2\bwlts{i}{\alpha}S_2 \Rightarrow
		\exists S_1, g, R_1\bwlts{j}{\alpha}S_1, g = f\frestr{\encmp{S_1}}, (S_1, S_2, g)\in\rel \label{hpb_rccs4}
	\end{align}

	We write that \emph{\(R_1\) and \(R_2\) are \HorHPB} if there exists a \HorHPB relation between them.

\end{definition}
Above, we write \(g\frestr{\encmp{R}}\)%
for the restriction of each component of \(g\) to \(\encmp{R}\). %
Note that the definitions above reflect the definition of \HorHPB (\autoref{def:hpb}): the condition \((\emptymem \vartriangleright \orig{R_1}, \emptymem \vartriangleright \orig{R_2}, \emptyset) \in \rel\) is intuitively the counterpart to the condition that \((\emptyset,\emptyset,\emptyset)\) has to be included in the relation on configuration structures.
Also, \(f\) shares similarity with the \lop bijection, in the sense that it exists, with \(\id\) as the component on the labels, iff there exists a \lop bijection between the unique maximal configurations in \(\encmp{R_1}\) and \(\encmp{R_2}\) (\autoref{lem:lop_preserve_conf}).

Relations defined on RCCS (\BF, \SBF, \HPB and \HHPB) straightforwardly extend to CCS, by simply stating that \(P_1\) and \(P_2\) are in it if \(\emptymem \vartriangleright P_1\) and \(\emptymem \vartriangleright P_2\) are too.
Therefore we can state below our results in terms of CCS processes. %

\begin{restatable}[Equivalences]%
	{theorem}{hpbrccs}
	\label{th:hpb_rccs}
	\(P_1\) and \(P_2\) are \HHPB (\resp \HPB) iff \(\enc{P_1}\) and \(\enc{P_2}\) are.
\end{restatable}

This result can easily be extended to \emph{weak}-\HPB and \emph{weak}-\HHPB~\cite{Bednarczyk1991, Phillips2014}, which are defined by removing from \HorHPB on configuration structures (\autoref{def:hpb}) and on RCCS (\autoref{def:hpb_rccs}) the condition that \(f\) must be preserved from one step to the next one. %

\begin{restatable}[Equivalence (contd)]{theorem}{bfhpbcoincide}
	\label{thm:def_hpb_coincide}
	\(P_1\) and \(P_2\) are \BF iff they are \HHPB.
\end{restatable}

\autoref{th:hpb_rccs} (\resp \autoref{thm:def_hpb_coincide}) uses our operational correspondence between RCCS processes (\resp RCCS memories) and their encodings as configuration structures~\cite[Lemma 6]{Aubert2016jlamp} (\resp as \(\iconf\)-structure~(\autoref{lem:transition_identified})) to transition between the semantic and syntactic worlds.

Our main result will come as an immediate corollary of Theorems~\ref{th:hpb_rccs} and \ref{thm:def_hpb_coincide}. %
\begin{corollary}[Main result]
	\label{cor:previousresult}
	\(P_1\) and \(P_2\) are \BF iff \(\enc{P_1}\) and \(\enc{P_2}\) are \HHPB.
\end{corollary}

\section{Concluding Remarks}
\label{sec:conclusion}

This work offers a \enquote{definitive} answer to the question of finding a meaningful bisimulation for reversible LTS by providing relations that correspond to \HorHPB on their encodings on \emph{all} processes.
We believe this contribution is of value because:
\begin{enumerate*}
	\item This result solves a problem that was open since \HHPB was defined~\cite{Bednarczyk1991}, nearly 30 years ago, for which despite the use of multiple techniques, only partial results were obtained,
	\item This idea in appearance simple still requires a lot of technical work, as detailed in the Appendix,
	\item The use of reversibility (both the backtracking capability \emph{and} the memory mechanism) is critical to characterize \HHPB on syntactical terms.
\end{enumerate*}
This result also enforces the importance of identifiers in general and not just as part of a backtracking mechanism.
Indeed, they are generally already present when concurrency is implemented, \eg when two \texttt{Unix} threads terminate with the same signal, the parent process have the capacity of determining which process sent which signal.

As a byproduct of our result, we also proposed an encoding of RCCS memories into an \enquote{enriched} configuration structure, called identified configuration structure.
This observation echoes our previous formalism~\cite{Aubert2016jlamp} and similar encoding~\cite{Graversen2018} in an interesting way:
as mentioned in the Introduction, a reversible process \(R\) was encoded as a pair \((\enc{\orig{R}}, x_{R})\) made of the configuration structure encoding the origin of \(R\), and a configuration \(x_{R}\) in it, called the address of \(R\).
The intuition was that we could \enquote{match} a partially executed process with a configuration. We can now go further by observing that \(\encmp{R}\) is isomorphic to the \(\iconf\)-structure
generated by \(x_R\), which is everything \enquote{below} it.
This result (\autoref{lem:enc_mem_2}) is used in our proof, and exemplified by \autoref{ex:proc2}: the encoding of the memory of \((\mem{2, a, 0} . \fork . \emptymem \rhd b) \mid (\mem{1,c,0} . \fork . \emptymem \rhd 0)\) corresponds the \enquote{past} of the process, whose underlying structure is grayed out in \autoref{fig:ex_mem4}, and what is left to execute---\(b \mid 0\)---corresponds to the \enquote{future} of that process, and is represented by the configuration \(\{c, a, b\}\) in \autoref{fig:ex_mem4}.

\bibliography{standalone} %

\appendix
\clearpage{}%
\section{Event Structures as Categories}%
\label{sec:ct-app}

Configuration structures
often use the insights provided by the
categorical framework~\cite{Aubert2016jlamp,Sassone1996,Winskel1982}.
This appendix regroups the categorical treatment of (identified) configuration structures (\autoref{def:conf_str}).%

\begin{definition}[Category of configuration structures]
	\label{def:cat_lcs}
	We define \(\catcs\) the category of configuration structures, where an object is a configuration structure, and a morphism \(f: \conf_1 \to \conf_2\) is a triple \((f_E, f_L, f_C)\) such that
	\begin{itemize}
		\item \(f_L : L_1\to L_2\);
		\item \(f_E : E_1 \to E_2\) preserves labels:
		      \(\labl_2(f_E(e)) = f_L(\labl_1(e))\);
		\item \(f_C : C_1 \to C_2\) is defined as
		      \(f_C(x) = \{ f_E(e) \setst e\in x\}\).
	\end{itemize}
	We write \(\conf_1 \iso \conf_2\) if there exists an isomorphism between \(\conf_1\) and \(\conf_2\).
\end{definition}
For simplicity, we often assume that \(L_1 = L_2\), \ie, that all the configuration structures use the same set of labels, take \(f_L\) to be the identity and remove it from the notation.

\begin{definition}[Category of \(\iconf\)-structures]
	\label{def:cat_ics}
	We define \(\catics\) the category of identified configuration structures, where objects are \(\iconf\)-structures, and a morphism \(f: \iconf_1 \to \iconf_2\) is a tuple \( q = (f, f_{\iden})\) such that
	\begin{itemize}
		\item \(f = (f_E, f_C)\) is a morphism in \(\catcs\) between the underlying structures of \(\iconf_1\) and \(\iconf_2\),%
		\item \(f_{\iden} : I_1 \to I_2\) preserves identifiers: \(f_{\iden}(\iden_1(e)) = \iden_2(f_E(e))\).
	\end{itemize}
	We write \(\iconf_1 \iso \iconf_2\) if there exists an isomorphism between \(\iconf_1\) and \(\iconf_2\).
\end{definition}
Observe that %
\(\catcs\) is a subcategory of \(\catics\).
In both \(\catcs\) and \(\catics\), composition is written \(\circ\) and defined componentwise.

\begin{lemma}
	\label{lem:cat_ics}
	Identified configuration structures and their morphisms form a category.
\end{lemma}

\begin{proof}
	\begin{description}
		\item[Identity] For every \(\iconf\)-structure \(\iconf = (E, C, \labl, I, \iden)\), \(\id_{\iconf} : \iconf \to \iconf\) is defined to be the identity on the underlying configuration structure \(\id : (E, C, \labl) \to (E, C, \labl)\) from \(\catcs\), that trivially preserves identifiers.
		      For any morphism \(f : \iconf_1 \to \iconf_2\), \(f \circ \id_{\iconf_1} = f = \id_{\iconf_2}\circ f\) is trivial.
		\item[Associativity] for \(f : \iconf_1 \to \iconf_2\), \(g : \iconf_2 \to \iconf_3\) and \(h : \iconf_3 \to \iconf_4\), \(h \circ (g \circ f) = (h \circ g) \circ f\) is inherited from the associativity in \(\catcs\), and since \(f\), \(g\) and \(h\) all preserves identifiers.
	\end{description}
	Hence \(\catics\) is a category.
\end{proof}

Unsurprisingly, a forgetful functor and an enrichment functor can be defined between those two categories.
The only assumption is that we need to suppose that every configuration structure can be endowed with a total ordering \(\preceq\) on its events.

\begin{lemma}
	\label{lem:functors}
	The forgetful functor \(\functoric : \catics \to \catcs\), defined by
	\begin{itemize}
		\item \(\functoric (E, C, \labl, I, \iden) = (E, C, \labl)\)
		\item \(\functoric(f_E, f_C, f_{\iden}) = (f_E, f_C)\)
	\end{itemize}
	and the enrichment functor \(\functorci : \catcs \to \catics\), defined by
	\begin{itemize}
		\item \(\functorci (E, C, \labl) = (E, C, \labl, I, \iden)\), where \(I = \{1, \hdots, \card{E}\}\) for \( \card{E}\) the cardinality of \(E\), and
		      \[
			      \iden(e) = \begin{cases*}
				      1 & if \(\forall e', e \preceq e'\) \\
				      i+1 & if \(\exists e', e' \preceq e\), \(\iden(e') = i\) and there is no \(e''\) \st \(e' \preceq e'' \preceq e\)
			      \end{cases*}
		      \]
		\item For \((f_E, f_C): (E_1, C_1, \labl_1) \to (E_2, C_2, \labl_2)\), \(\functorci(f_E, f_C) = (f_E, f_C, f_{\iden})\), where we let \(f_{\iden}(\iden_1(e)) = \iden_2(f_E(e_2))\).
	\end{itemize}
	are functors.
\end{lemma}
\begin{proof}
	Proving that \(\functoric\) is a functor is immediate.

	Proving that \(\functorci(\conf)\) is a \(\iconf\)-structure is immediate, since our construction of \(\iden\) trivially insures \ref{eq:collision}.
	For \((f_E, f_C): \conf_1 \to \conf_2\), proving that \(\functorci(f_E, f_C)\) is a morphism between \(\functorci(\conf_1)\) and \(\functorci(\conf_2)\) is also immediate.
	For the preservation of the identity, we compute:
	\begin{align*}
		\functorci(\id_{\conf}) & = \functorci(\id_E, \id_C)   \\
		                        & = (\id_E, \id_C, f_{\iden})  \\
		\shortintertext{where \(f_{\iden} (\iden (e)) = \iden(\id_E(e)) = \iden(e)\), hence \(f_{\iden} = \id_{I} : I \to I\),}
		                        & = (\id_E, \id_C, \id_{\ids}) \\
		                        & = \id_{\functorci(\conf)}
	\end{align*}
	For the composition of morphisms, given \(f = (f_C, f_E) : \conf_1 \to \conf_2\) and \(g = (g_C, g_E) : \conf_2 \to \conf_3\), we write \(\functorci(\conf_i) = (E_i, C_i, \labl_i, I_i, \iden_i)\) and we compute:
	\begin{align*}
		\functorci(g) \circ \functorci(f) & = (g_C, g_E, g_{\iden}) \circ (f_C, f_E, f_{\iden})         \\
		                                  & = (g_C \circ f_C, g_E \circ f_E, g_{\iden} \circ f_{\iden}) \\
		\shortintertext{where, for all \(e \in E_1\), we compute:
			\begin{align*}
				(g_{\iden} \circ f_{\iden})(\iden_1(e)) & = g_{\iden} (f_{\iden}(\iden_1(e))                                            \\
				                                        & = g_{\iden} (\iden_2(f_E(e))) \tag{Since \(f_{\iden}\) preserves identifiers} \\
				                                        & = \iden_3(g_E (f_E(e))) \tag{Since \(g_{\iden}\) preserves identifiers}
			\end{align*}
			Hence we can conclude:
		}
		\functorci(g) \circ \functorci(f) & = \functorci(g \circ f)\qedhere
	\end{align*}\end{proof}

\begin{remark}
	\label{rem:mor-in-catics-notation}
	In \(\catics\), every morphism \(f = (f_E, f_L, f_C, f_{\iden})\) from \(\iconf_1\) to \(\iconf_2\) is actually fully determined by \(f_E\) whenever \(f_L = \id\).
	Indeed, given \(f_E : E_1 \to E_2\), then we can define for all \(x \in C_1\), \(f_C(x) = \{ f_E(e) \setst e \in x\}\) and \(f_{\iden}\) as \(f_{\iden}(\iden_1(e)) = \iden_2(f_E(e))\).
	We will often make the abuse of notation of writing \(f_E\) for \(f\) and reciprocally.
\end{remark}

\section{Concurrency in a Trace and Trace Equivalence}
\label{sec:trace}

We give here a quick reminder on concurrency and causality in CCS~\cite{Boudol1988} and RCCS~\cite{Danos2004} traces.
Aside from the convenient notation \(m_{R/S}\) that represents the memory stack(s) modified by a forward transition from \(R\) to \(S\), and of the notation \(\namelist{a}\) for a list of names \(a_1,\cdots,a_n\), nothing new is introduced in this Section.
However, the results reminded below are used in the proofs of \autoref{lem:transition_identified} and \autoref{th:equivalence_bf}.

Concurrency on events corresponds to a notion of concurrency on transitions in RCCS traces%
~\cite[Definition 7 and Lemma 8]{Danos2004}.
For this reminder we consider only concurrency and causality for forward transitions, so that CCS intuitions work equally well.
We make a remark at the end about extending the concurrency to backward transitions, but it should be noted that forward and backward transitions are not mixed.

Two transitions \(t_1 = R\fwlts{i_1}{\alpha_1} R_1\) and \(t_2 = R'\fwlts{i_2}{\alpha_2} R_2\) are \emph{composable} if \(R_1 = R'\), and in this case, doing \(t_1\) then \(t_2\) is written as the composition \(t_1 ; t_2\).
Given \(n\) composable transitions \(t_i : R_i \fwlts{i}{\alpha_i} R_{i+1}\) and their composition \(t_1 ; \hdots ; t_n\), we say that \emph{\(t_i\) is a direct cause of \(t_k\)} for \(1 \leqslant i < k \leqslant n\) and write \(t_i < t_k\) (or, for short, \(i < k\)) if there is a memory stack \(m\) in \(R_{i+1}\) and a memory stack \(m'\) in \(R_{k+1}\) such that \(m < m'\), where the order on memory stacks is given by prefix ordering.
Note that, if they exist, \(m\) and \(m'\) are unique, as memory events in reachable processes all have a different pairs (identifier, label).

Let \(R\fwlts{i}{\alpha} S\) be a transition.
If \(\alpha\neq\tau\), we write \(m_{R/S} = \{m\}\) where %
\begin{align*}
	R = & (\cdots((R_3 \mid ((R_1 \mid (m \rhd P)) \mid R_2)\bs \namelist{b^1}) \mid R_4)\bs \namelist{b^2} \cdots \mid R_n) \bs \namelist{b^{m}}                 \\
	S = & (\cdots ((R_3 \mid ((R_1 \mid \mem{i, a, Q} . m \rhd P ) \mid R_2)\bs \namelist{b^1}) \mid R_4)\bs \namelist{b^2} \cdots \mid R_n) \bs \namelist{b^{m}}
\end{align*}
for some \(R_i\) any of which could be missing and for some \(\namelist{b^j}\), possibly missing as well.
If \(\alpha = \tau\), then \(m_{R/S}\) will contain the pair of memory stacks that has been changed by the transition.
Intuitively, the notation \(m_{R/S}\) is useful to extract the memory stack(s) modified by a forward transition from \(R\) to \(S\).

Two transitions are \emph{coinitial} if they have the same source process and \emph{cofinal} if they have the same target process.
We say that two coinitial transitions \(t_1 = R\fwlts{i_1}{\alpha_1} S_1\) and \(t_2 = R\fwlts{i_2}{\alpha_2} S_2\) are \emph{concurrent} if \(m_{R/S_1} \cap m_{R/S_2}=\emptyset\), that is, if the transitions modify disjoint memories in \(R\).

The square lemma~\cite[Lemma 8]{Danos2004} says that
moreover, given two such concurrent transitions,
there exists two cofinal and concurrent transitions \(t_1' = S_1\fwlts{i_2}{\alpha_2} S\) and \(t_2' = S_2\fwlts{i_1}{\alpha_1} S\).
The name of the lemma comes from this picture:

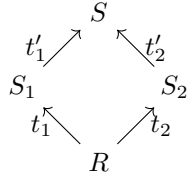
\begin{center}
	\begin{tikzpicture}
		\node (R) at (0,0) {\(R\)};
		\node (S1) at (-1, 1) {\(S_1\)};
		\node (S2) at (1, 1) {\(S_2\)};
		\node (S) at (0, 2) {\(S\)};
		\draw [->] (R) -- node[left]{\(t_1\)} (S1);
		\draw [->] (R) -- node[right=2]{\(t_2\)} (S2);
		\draw [->] (S1) -- node[left=2]{\(t'_1\)} (S);
		\draw [->] (S2) -- node[right]{\(t'_2\)} (S);
	\end{tikzpicture}
\end{center}

Moreover, the traces \(\theta_1 = t_1;t_1'\) and \(\theta_2 = t_2;t_2'\) are \emph{equivalent}~\cite[Definition 9]{Danos2004}. This allows one to define \emph{equivalence classes on transitions}: \(t_1\) in \(\theta_1\) is equivalent to \(t_2'\) in \(\theta_2\) if \(\theta_1\) is equivalent to \(\theta_2\) and \(t_1\) and \(t_2'\) have the same index. Then in the trace \(t_1;t_1'\) we are now allowed to say that \(t_1\) is concurrent to \(t_1'\).

In a trace \(t_1;t_2\) we have that \(t_1\) is concurrent to \(t_2\) iff \(t_1\) is not a cause of \(t_2\).
This follows from a case analysis using the definitions of concurrency and causality.
Thanks to trace equivalence, we also have that in a trace \(t_1;\hdots;t_n\) either \(t_1\) is a cause of \(t_n\) or the two transitions are concurrent.
Those intuitions are enough for us to carry on our development, but a complete treatment of concurrency and causality in the trace of a CCS process~\cite{Boudol1988} can give better insight to the curious reader.

The definitions of concurrency for forward coinitial traces and of causality for forward traces can easily be \enquote{flipped} into definitions of concurrency for backward cofinal traces, and of causality for backward traces.

\section{Proofs and Auxiliary Materials}
\label{sec:proof-and-aux}

In this section we detail the proofs missing from the main text, and introduce some intermediate definitions and lemmas that are necessary for the proofs.

\subsection{Operations on Identified Configurations Structures (\autoref{sec:op-def})}%
\label{sec:op-proof}

Our main goal here is to prove that the operations of \autoref{icat-op-def} preserve \(\iconf\)-structures (\autoref{lem:preserve_ics}).
The product and coproduct (used to define the nondeterministic choice below) have particular roles, since they have a direct representation in the categorical world.

The structures we considered are \emph{full} \wrt the sets of labels and identifiers, \ie the labeling and identifying functions are surjective. This only impacts the relabelling and reidentifying operations, where we have to additionally require that \(\labl'\) and \(\iden'\) are surjective.

We redefine the nondeterministic choice of \autoref{icat-op-def} by first defining the coproduct on \(\iconf\)-structures and then using relabeling and reidentifying to get rid of the extra indices in the label and the identifier of events.
It is easy to check that the two definitions of nondeterministic choice are equivalent, but working with the one from \autoref{icat-op-def} is easier and simpler.

\begin{definition}
	\label{def:coproduct}

	We redefine nondeterministic choice using coproduct as follows:

	\begin{description}
		\item[The coproduct] of \(\iconf_1\) and \(\iconf_2\) is \(\iconf_1 \pm \iconf_2 = (E, C, L,\labl, I,\iden)\), where
			      {
				      \setlength\multicolsep{\topsep}
				      \begin{multicols}{2}
					      \begin{itemize}
						      \item \(E=\{\{1\}\times E_1\}\cup\{\{2\}\times E_2\}\) with \(\pi_1 : E \to \{1,2\}\) and \(\pi_2: E \to E_1\cup E_2\),
						      \item \(C=\{\{i\}\times x \setst x\in C_i\}\),
						      \item \(L = \{\{1\} \times L_1\} \cup \{\{2\} \times L_2\}\),
						      \item \(\labl(e)=(i,\labl_i(\pi_2(e))\)) for \(\pi_1(e) = i\),
						      \item \(I =\{\{1\}\times I_1\}\cup\{\{2\}\times I_2\} \),
						      \item \(\iden(e) = (i, \iden_i(\pi_2(e)))\) for \(\pi_1(e) = i\).
					      \end{itemize}
				      \end{multicols}
				      and with the expected injections \(\iota_i: \iconf_i \to \iconf_1+\iconf_2 \).
			      }
		\item[The nondeterministic choice] of \(\iconf_1\) and \(\iconf_2\) is \(\iconf_1 + \iconf_2 = (\iconf_1+\iconf_2)[\labl'/\labl][\iden'/\iden]\) where
		      \begin{itemize}
			      \item \(\labl'(e) = a\text{ if }\labl(e)=(j, a), j\in\{1,2\}\),
			      \item \(\iden'(e) = i\text{ if }\labl(e)=(j, i), j\in\{1,2\}\).
		      \end{itemize}
	\end{description}
\end{definition}

\begin{lemma}
	\label{lem:product_ics}
	The product and coproduct of \(\iconf\)-structures is the product and coproduct in \(\catics\).
\end{lemma}

In the categorical setting, the product and coproduct on labeled configuration structures can be obtained by a straightforward enrichment of the un-labeled configuration structures~\cite[Propositions 11.2.2 and 11.2.3]{Winskel1995}.
In a similar vein, we obtain the extension of those operations on identified (labeled) configuration structures directly.

\begin{proof}
	The proof for the product and the coproduct follows the same pattern: assume \(\iconf\) is the result of applying the operation to \(\iconf_1\) and \(\iconf_2\), and observe that
	\begin{enumerate}
		\item \(\iconf\) is an \(\iconf\)-structure. This follows from the fact that the structure underlying \(\iconf\), \(\functoric(\iconf) = \conf\) is the product or coproduct of the configurations underlying \(\functoric(\iconf_1) = \conf_1\) and \(\functoric(\iconf_2) = \conf_2\), hence, that it is a valid configuration structure. As \(\iden_1\) and \(\iden_2\) are such that \ref{eq:collision} is enforced, then by construction \(\iden\) enforces it too and \(\iconf = \conf \oplus \iden\) is an \(\iconf\)-structure.

		\item The morphisms \(\gamma_i\) (for the product) and \(\iota_i\) (for the coproduct) extends the corresponding morphisms on configuration structure with valid morphisms on identifiers.

		      Let us suppose that \(\conf\) is the product and let \(\alpha_1:\conf\to\conf_1\) be one of the projections. Let us write \(\iconf = \conf \oplus \iden = (\conf, \iden, I)\) and \(\iconf_1 = \conf_1 \oplus \iden_1 = (\conf_1, \iden_1, I_1)\). From the definition of product on \(\iconf\)-structures, \(I = I_1 \times_{\star} I_2\) is the product on the sets on identifiers and let \(p_1:I \to I_1\) be one of the projections, defined as \(p_1(i_1,i_2) = i_1\), for \(i_1\neq\star\). We can now write \(\gamma_1 = (\alpha_1, p_1)\) the projection on \(\iconf\)-structures, which is indeed a morphism: suffices to verify that \(p_1(\iden(e)) = \iden_1(\alpha_1(e))\), for all events \(e\), which follows from the definition of \(\iden\). A similar argument holds for the coproduct and the injections.

		\item The universal properties follow easily by lifting the universal property of the underlying configuration structures in \(\catics\).\qedhere
	\end{enumerate}
\end{proof}

The restriction of the operations of \autoref{icat-op-def} to configuration structures are standard~\cite{Winskel1982,Winskel1986}, except %
for postfixing.
We prove below that the restriction of this operation to configuration structures is correct.

\begin{lemma}
	\label{lem-postfixing}
	The postfixing of a label \(a\) to an event structure \(\conf_1 = (E_1, C_1, L_1, \labl_1)\), defined as \(\conf_1 \post (a) = (E, C, L, \labl)\) where
	\begin{itemize}
		\begin{minipage}{0.65\linewidth}
			\item \(E = E_1 \cup \{e\}\), for \(e \notin E_1\),
			\item \(C = C_1 \cup \{x \cup \{e\} \setst x \in C_1 \text{ is maximal and finite}\}\),
		\end{minipage}
		\begin{minipage}{0.34\linewidth}
			\item \(L = L_1\cup\{a\}\),
			\item \(\labl = \labl_1 \cup \{e \mapsto a\}\),
		\end{minipage}
	\end{itemize}
	is a configuration structure.
\end{lemma}

\begin{proof}
	Looking back at \autoref{def:conf_str}, we simply have to prove that the four axioms of configuration structures are respected by \(\conf_1 \post (a)\), knowing that they are respected by \(\conf_1\) by assumption.
	\begin{itemize}
		\item \ref{Finiteness} is satisfied because \(\conf_1\) is a configuration structure, and every configuration in \(\{x \cup \{e\} \setst x \in C_1 \text{ is maximal and finite}\}\) is finite.
		\item For \ref{CoincidenceFreeness}, we only have to check the configurations containing \(e\), otherwise it folows from \(\conf_1\) being a configuration structure.
		      Given \(y \in \{x \cup \{e\} \setst x \in C_1 \text{ is maximal and finite}\}\), there exists \(y'\) such that \(y = y' \cup \{e\}\).
		      Given two events \(e_1, e_2\) in \(y\) such that \(e_1 \neq e_2\), if \(e_1 = e\) or \(e_2 = e\), then \(y'\) is the configuration we are looking for. Otherwise, it follows from \(y'\) being a configuration in the configuration structure \(\conf_1\).

		\item In \ref{FiniteCompleteness}, let \(X\) be a subset of configurations.
		      If \(\forall x \in X\), \(e \notin x\), then the result follows from \(\conf_1\) being a configuration structure.
		      Otherwise, there exists \(x' \in X\) such that \(e \in x'\), hence we know that \(x'\) is maximal.
		      Now, assume there is \(y \in C\) finite such that \(\forall x \in X\), \(x \subseteq y\). As \(x'\) is finite and maximal and \(y\) is finite, then \(x' = y\). Suppose by contradiction that \({\textstyle \bigcup } X \notin C\). It implies that \(\exists y' \in X\) with \(e' \in y'\) such that \(e'\notin x'\). We reach a contradiction: as \(y' \subseteq y\), \(e' \in y\), and as \(x' \subseteq y\) and \(e'\notin x'\), we obtain that \(x' \subsetneq y\), which contradicts the maximality of \(x'\).
		      Hence, for all \(x \in X\), \(x \subseteq x'\), and \({\textstyle \bigcup } X = x' = y\).
		\item For \ref{Stability} let us consider \(x\) and \(y\) such that \(x\cup y \in C\). If \(e\notin x\cup y\) the result follows from the stability of \(\conf_1\). Otherwise assume, that \(e\in x\), but then \(x\) is maximal,
		      and for \(x \cup y \in C\) to be the case it must be that either \(y = \emptyset\) or \(y \subseteq x\), and \(x \cap y \in C\) holds trivially.\qedhere
	\end{itemize}
\end{proof}

\begin{lemma}
	\label{lem:preserve_ics}
	The operations of \autoref{icat-op-def} (relabeling, reidentifiying, restriction, prefixing, postfixing, non-deterministic choice and product), coproduct, as well as the parallel composition (\autoref{def:par-ics}) preserve \(\iconf\)-structures.
\end{lemma}

\begin{proof}
	Let us note that \begin{enumerate*} \item the restriction of the relabeling, restrictions, prefixing, postfixing, non-deterministic choice, coproduct and product to configuration structures preserve configuration structures, as their are adaptations of the usual operations~\cite{Winskel1982,Winskel1986}, or proven correct in \autoref{lem-postfixing}, \item any configuration structure endowed with a valid identifying function (\ie, such that no two events in the same configuration have the same identifier, \cf \ref{eq:collision}) is a valid \(\iconf\)-structure.\end{enumerate*}
	\begin{description}
		\item[For the relabeling,] since this operation does not change anything but the labels, we have nothing to prove.
		\item[For the reidentifiying,] since the function \(\iden'\) is supposed to make the resulting \(\iconf\)-structure respect the \ref{eq:collision} condition, there is nothing to prove.
		\item[For the restriction,] note that this operation only removes events in configurations and keeps the identifying function intact. Hence if the initial structure has a valid identifying function, then the identifying function of the new structure is a valid one by assumption.
		\item[For the prefixing,] since \(i\) is a fresh identifier, \ref{eq:collision} is trivially preserved.
		\item[For the postfixing,] we know that the underlying configuration of \(\iconf_1 \post (a, i)\), \(\conf_1 \post (a)\) is a valid configuration structure by \autoref{lem-postfixing}.
		      And since \(i\) is a fresh identifier, \ref{eq:collision} is trivially preserved.
		\item[For the coproduct,] it follows trivially from \autoref{lem:product_ics}.
		\item[For the product,] it follows trivially from \autoref{lem:product_ics}.
		\item[For the nondeterministic choice,] note that \ref{eq:collision} holds only if \(I_1\cap I_2 = \emptyset\). We assume this to be the case, as we can always apply a reidentifying operation when it is not, to guarantee that nondeterministic choice is well defined. Moreover the two definitions of nondeterministic choice from \autoref{icat-op-def} and \autoref{def:coproduct} coincide, which follows trivially.
		\item[For the parallel composition,] we need to first observe that this operation is defined in terms of product, restriction, relabeling, and reidentifiying along \(\iden'\), and we know that those operations are correct provided \(\iden'\) respects the \ref{eq:collision} condition.
		      In other terms, given \(\iconf_1\), \(\iconf_2\) and their product \(\iconf_3 = (E_3, C_3, L_3,\labl_3, I_3,\iden_3)\), we only need to prove that \(\iconf_3 [\iden'/\iden_3]\) is a valid \(\iconf\)-structure, and it follows from a case analysis.
		      Given a configuration \(x \in C_3\), as and two events \(e,e' \in x\), we can reason by case analysis on the first and second projections of both events, and these are the possible cases:
		      \begin{itemize}
			      \item \(\pi_2(e)= \star\) and \(\pi_2(e')= \star\).
			            In this case, looking at the definition of the product in \(\iconf\)-structures, \(\iden'(e) = (\iden_1(\pi_1(e)), \star)\) and \(\iden'(e') = (\iden_1(\pi_1(e')), \star)\).
			            If \(\iden'(e) = \iden'(e')\), then \(\iden_1(\pi_1(e)) = \iden_1(\pi_1(e'))\) in the configuration \(\pi_1(x)\) in \(\iconf_1\).
			            But that's a contradiction, since \(\pi_1(e)\) and \(\pi_1(e')\) are in the same configuration in \(\iconf_1\) and the identifying function of \(\iconf_1\) is valid.
			      \item \(\pi_1(e)= \star\) and \(\pi_1(e')= \star\). This case is similar as the previous one, except that it uses that the identifying function of \(\iconf_2\) is valid.
			      \item \(\pi_1(e)\neq \star\), \(\pi_2(e)= \star\) and \(\pi_2(e')\neq \star\) (with either \(\pi_1(e') = \star\) or \(\pi_1(e')\neq \star\)).
			            If \(\iden'(e) = \bot_{(\iden_1(\pi_1(e)), \star)}\), then clearly \(\iden'(e) \neq \iden'(e')\) as \(\pi_2(e')\neq \star\).
			            If \(\iden'(e) = \iden_1(\pi_1(e))\), then we know that \(\forall e_2 \in E_2, \iden_2(e_2) \neq \iden_1(\pi_1(e))\), and it is impossible that \(\iden'(e') = \iden'(e)\), as \(\pi_2(\iden'(e')) = \iden_1(\pi_1(e))\) cannot be.
			      \item \(\pi_1(e)\neq \star\), \(\pi_2(e)\neq \star\) and \(\pi_2(e')\neq \star\) (again with either \(\pi_1(e') = \star\) or \(\pi_1(e')\neq \star\)).
			            If \(\iden'(e) = \bot_{\iden_3(e)}\), then having \(\iden'(e) = \iden'(e')\) would imply \(\iden'(e') = \bot_{\iden_3(e')}\) and \(\iden_3(e) = \iden_3(e')\), which would contradict that product is well-defined.
			            Otherwise, \(e\) is a synchronization or a fork (\cf \ref{valid-synch} in \autoref{def:par-ics}), and it follows that \(\iden_1(\pi_1(e)) = \iden_2(\pi_2(e))\) and \(\iden'(e) =\iden_1(\pi_1(e))\).
			            We prove by case analysis that \(\iden'(e) = \iden'(e')\) can not be the case:
			            \begin{itemize}
				            \item If \(\pi_1(e') = \star\), then \(\iden'(e') = \bot_{\iden_3(e')}\) as \(\pi_1(e)\) prevents from applying \ref{extra2}.
				            \item If \(\pi_1(e')\neq \star\), \(\iden_3(e') = \iden_3(e)\), which would contradict the correctness of the product.
			            \end{itemize}
			      \item \(\pi_1(e)\neq \star\), \(\pi_2(e)\neq \star\) and \(\pi_2(e')= \star\) is similar to above, except that it uses \ref{extra1}.
			      \item Finally, if \(\pi_1(e)\neq \star\), \(\pi_2(e)\neq \star\), \(\pi_1(e') \neq \star\), \(\pi_2(e') \neq \star\), then \(\iden'(e) = \iden(e')\) would require to have applied either \ref{valid-synch} or \ref{error} to both cases, and both situations would contradict the correctness of the product.
			            \qedhere
		      \end{itemize}
	\end{description}
\end{proof}

\subsection{Properties of Memory Encodings and Operational Correspondence (Sections~\ref{sec:enc_mem} and \ref{sec:op-corr})}
\label{sub:prop-mem-proof}

Our goal here is to show that there is an operational correspondence between \(R\) and \(\encmp{R}\) (\autoref{lem:transition_identified}, \autoref{sssec:op_cor}), and this requires intermediate lemmas (Lemmas~\ref{lem:equivalence_encoding}--\ref{lem:prevent-max-event}%
) about the encoding of memories and their relation to maximal events.
To prove those, we start by exhibiting some useful properties of memory encoding (\autoref{sec:mem-prep}).

\subsubsection{Properties of Memory Encodings}
\label{sec:mem-prep}
We assume
given reachable processes
\(R\) and \(S\) %
and we write \(\orig{R}\) for the origin of \(R\), and \(\encmp{R}\) as \((E_R, C_R, \labl_R, I_R, \iden_R)\) and similarly for \(S\).

To prove interesting properties about the encoding of memory, we first need this small technical lemma.

\begin{lemma}
	\label{lem:unique_id}
	For every reversible thread \(m\rhd P\) of a reachable process \(R\), and for all \(i \in \ids(m)\), \(i\) occurs once in \(m\).
\end{lemma}

\begin{proof}
	We prove it by induction on the structure of \(m\).
	\begin{itemize}
		\item if \(m = \emptyset\), then it is obvious.
		\item if \(m = \fork . m'\), then by induction hypothesis, for all \(i \in \ids(m')\), \(i\) occurs only once in \(m'\), and since no identifier occur in \(\fork\), \(i\) occurs only once in \(m\).
		\item if \(m = \mem{j, \lambda, Q}.m'\) then there exists \(S\) such that \(m' \rhd S \fwlts{j}{\lambda} \mem{j, \lambda, Q}.m' \rhd P\) since \(R\) is reachable and \(\mem{j, \lambda, Q}.m' \rhd P\) is a thread in it.
		      By induction we know that for all \(i \in \ids(m')\), \(i\) occurs once in \(m'\).
		      We reason on the derivation tree of the transition that adds the memory event \(\mem{j, \lambda, Q}\) and we have that for any such transition, the rule act.\@ of \autoref{fig:ltsrules} is applied as axiom at the top of the derivation tree.
		      By the side condition of the rule act., we know that \(j \notin \ids(m')\), hence that for all \(i \in \ids(m)\), \(i\) occurs once in \(m\). \qedhere
	\end{itemize}
\end{proof}

Note that the property above holds for reversible threads, and not for RCCS processes in general: we actually \emph{want} memory events to sometimes share the same identifiers.
Indeed, two memory events need to have the same identifiers if they result from a synchronization (\ie, the application of the syn.\@ rule of \autoref{fig:ltsrules}) or a fork (\ie, the application of the \ref{eq:rcss-congru-distrib-mem} rule of structural equivalence, \autoref{def:congru_rccs}).

\begin{restatable}[Uniqueness of identifiers]{lemma}{lem:id_unique}
	\label{lem:id_unique}
	For all \(e_1, e_2 \in E_R\), \(\iden_R(e_1) = \iden_R(e_2)\) implies \(e_1 = e_2\).
\end{restatable}

\begin{proof}
	We proceed by structural induction on \(R\).
	From \autoref{lem:unique_id} the only interesting case is the parallel composition, \ie \(R= R_1\mid R_2\).
	From the definition of parallel composition of \(\iconf\)-structures (\autoref{def:par-ics}), it follows that \(\iden_R(e_1) = \iden_R(e_2)\) implies \(e_1 = e_2\).
\end{proof}

\prefixenc*

\begin{proof}
	We proceed by induction on \(R\).

	\begin{description}
		\item[If \(R\) is \(m \rhd P\),] we prove that \(\encmp{m \rhd P}\) is a partially ordered set (poset) with one maximal element by induction on \(m\).
		      The base case, if \(m\) is \(\emptymem\), is trivial, since \(\encm{\emptymem} = \iconfzero \) is a poset with one maximal element, \(\emptyset\).
		      If \(m\) is \(\fork . m'\), then it follows by induction hypothesis, since \(\encm{\fork . m'} = \encm{m'}\).
		      If \(m\) is \(\mem{i, \alpha, P} . m'\), then by induction hypothesis, \(\encm{m'}\) is a poset with one maximal element, and the postfixing construction used to define \(\encm{\mem{i, \alpha, P} . m'}\), detailed in \autoref{icat-op-def}, preserves that property.

		\item[If \(R\) is \(R' \bs a\),] then it trivially follows by induction hypothesis.

		\item[If \(R\) is \(R_1 \mid R_2\),] then by induction hypothesis we get that \(\encmp{R_1}\) and \(\encmp{R_2}\) are both posets with a maximal configuration.
		      We also know by \autoref{lem:id_unique} that events in \(\encmp{R_1}\) and \(\encmp{R_2}\) have disjoint identifiers (which does not imply that identifiers occurring in \(\encmp{R_1}\) and in \(\encmp{R_2}\) need to be disjoint).
		      Looking at the definition of parallel composition for \(\iconf\)-structures (\autoref{def:par-ics}), we may observe that \(\encmp{R} = \encmp{R_1 \mid R_2} = \encmp{R_1} \mid \encmp{R_2}\) consists of the structure \(\big(\encmp{R_1}\times\encmp{R_2})[\iden'/\iden][\labl'/\labl] \big)\crestr{\bot}\) for a certain \(\iden'\) and \(\labl'\).

		      We show that there exists more than one maximal configurations in \(\encmp{R_1}\times\encmp{R_2}\) and that all but one are removed by the restriction.

		      We show this by first showing that there exists more than one maximal configurations in \(\encmp{R_1}\times\encmp{R_2}\), denoted here with \(\iconf\).
		      From the definition of product (\autoref{icat-op-def}), we have that there exists \(y_1,\hdots, y_n\) maximal configurations in \(\iconf\) such that \(\pi_1(y_i)\) and \(\pi_2(y_i)\) are maximal in \(\iconf_1\) and \(\iconf_2\), respectively.

		      A second step is then to show that the restriction keeps only one maximal configuration.
		      Let \(y_i\) and \(y_j\) be two maximal configurations in \(\encmp{R_1}\times\encmp{R_2}\).
		      As they are maximal it implies that \(y_i\cup y_j\) is not a configuration in \(\iconf\), for \(i\neq j\leqslant n\).
		      In turn, this implies that there exists \(e_i\in y_i\) and \(e_j\in y_j\) such that \(\pi_1(e_i) = \pi_1(e_j)\) and \(e_i\neq e_j\), as otherwise \(y_i\cup y_j\) would be defined.
		      Here we assume that \(\pi_1(e_i) = \pi_1(e_j)\) but we could also take \(\pi_2(e_i) = \pi_2(e_j)\) and the argument still holds.

		      Let us now take \(d\) an event in \(\encmp{R_1}\) and take \(e_1, \hdots, e_m\) the subset of events in \(E\) such that \(\pi_1(e_i) = d\). The restriction in the parallel composition of \(\encmp{R_1} \mid \encmp{R_2}\) keeps only one such event \(e_i\) and removes the rest: all the other events are reidentified with \(\bot_k\) since no two events can have the same identifier in the same component of the event.
		      Therefore, from all maximal configurations \(y_1, \hdots, y_m\) such that \(e_i \in y_i\), \(i \leqslant m\), only one remains.

		      By applying the argument above to all events in \(\encmp{R_1}\) (and \(\encmp{R_2}\)), we have that the restriction removes all but one \(y_i\), which is then the maximal configuration in \(\encmp{R_1} \mid \encmp{R_2}\).\qedhere
	\end{description}
\end{proof}

\subsubsection{Operational Correspondence}
\label{sssec:op_cor}

\begin{lemma}
	\label{lem:equivalence_encoding}
	If \(R\congru S\), then the exists an isomorphism \(f\) between \(\encmp{R}\) and \(\encmp{S}\), with \(f_L = \id\) and \(f_{\iden}=\id\). %
\end{lemma}

\begin{proof}
	Looking back at \autoref{def:congru_rccs}, there are only limited ways for \(R\) and \(S\) to be structurally equivalent (\ie, in the \(\congru\) relation), and we review them one by one.
	Before that, let us observe that the sets of identifiers and of labels occurring in the memories of \(R\) and \(S\) are identical.

	In%
	~\ref{eq:rcss-congru-alpha}, it should be noted that the memory is left untouched, so their encodings are equal.

	For~\ref{eq:rcss-congru-distrib-mem}, we have that
	\begin{align*}
		\encmp{m \rhd (P \mid Q)}                         & = \encm{m}                                                  \\
		\shortintertext{and}
		\encmp{(\fork. m \rhd P) \mid (\fork . m \rhd Q)} & = \encmp{(\fork. m \rhd P)} \mid \encmp{(\fork . m \rhd Q)} \\
		                                                  & = \encm{\fork. m} \mid \encm{\fork. m}                      \\
		                                                  & = \encm{m} \mid \encm{m}
	\end{align*}

	The construction of the isomorphism between \(\encm{m}\) and \(\encm{m} \mid \encm{m}\) maps \(e\) and \((e, e)\): let \(e\) be an event in \(\encm{m}\).
	Note that the only event \(e'\) in \(\encm{m} \mid \encm{m}\) such that \(\pi_1(e') = e\) is \((e, e)\).
	Indeed, suppose \((e, e'')\) is in \(\encm{m} \mid \encm{m}\), for \(e'' \neq e\).
	Then, by \autoref{lem:id_unique}, \(\iden(e) \neq \iden(e'')\), since they are both events in \(\encm{m}\).
	But by the definition of parallel composition of \(\iconf\)-structures (\autoref{def:par-ics}), \((e, e'')\) should have been re-identified with \(\bot_k\) and then discarded. %
	Similarly, \((e, \star)\) can not be an event in \(\encm{m} \mid \encm{m}\), since it should have been re-identified with \(\bot_k\) and then discarded. %
	Hence, we can map \(e\) in \(\encm{m}\) and \((e, e)\) in \(\encm{m} \mid \encm{m}\) and obtain our isomorphism.
	It follows from the definition of parallel composition of \(\iconf\)-structures that the functions \(f_L\) and \(f_{\iden}\) in the isomorphism are indeed the identities.

	For~\ref{eq:rcss-congru-scope}, we have that have \(\encmp{m \rhd (P \bs a)} = \encm{m} = \encmp{(m \rhd P) \bs a}\).
\end{proof}

\begin{lemma}
	\label{lem:postfix-max-event}
	The event introduced in the postfixing of a memory event to an identified structure is maximal in the resulting identified structure.
\end{lemma}

\begin{proof}
	The proof is immediate: the event introduced occurs only in the finite maximal configurations of the identified structure, and hence there cannot be any other event that causes it.
	If the maximal configuration is infinite, then the event introduced by the postfixing is not in it, and there cannot be any other event causing it.
\end{proof}

Furthermore, the maximality of an event can be \enquote{preserved} by parallel composition:

\begin{lemma}
	\label{lem:prevent-max-event}
	For all identified structure \(\iconf_1 = (E_1, C_1, L_1, \labl_1, I_1, \iden_1)\) with \(e_1 \in E_1\) a maximal event in it, and for all identified configuration \(\iconf_2 = (E_2, C_2, L_2, \labl_2, I_2, \iden_2)\) such that \(\iden_1(e_1) \notin I_2\), \((e_1, \star)\) is maximal in \(\iconf_1 \mid \iconf_2\).
\end{lemma}

\begin{proof}
	The definition of parallel composition of identified configuration structures (\autoref{def:par-ics}) should make it clear that the only event in \(\iconf_1 \mid \iconf_2\) whose first projection is \(e_1\) is \((e_1, \star)\), since \(\iden_1(e_1) \notin I_2\): all the other pairings of events from \(E_2\) with \(e_1\) being reidentified with \(\bot_k\) and subsequently removed. %
	Now, suppose for the sake of contradiction that \((e_1, \star)\) is not maximal in \(\iconf_1 \mid \iconf_2\).
	It means that there is a maximal configuration \(x\) in \(\iconf_1 \mid \iconf_2\) and an event \(e \in x\) such that \((e_1, \star) <_x e\).

	By the definition of product of \(\iconf\)-structures (\autoref{icat-op-def}), \(\pi_1(x) \in C_1\), and we prove that \(\pi_1(x)\) is maximal in \(\iconf_1\) by contradiction.
	Suppose \(\pi_1(x)\) is not maximal in \(\iconf_1\), then there exists \(z \in C_1\) such that \(\pi_1(x) \subsetneq z\).
	Assume \(z\) is maximal in \(\iconf_1\) (if it is not, then take \(z'\) the maximal configuration such that \(z \subseteq z'\), which always exists)%
	, and note that \(e_1 \in z\), since \(e_1 \in \pi_1(x)\) and \(\pi_1(x) \subset z\).
	By \ref{Stability}, since \(z \cup \pi_1(x) = z\) is a configuration in \(C_1\), \(z \cap \pi_1(x) = z \setminus \pi_1(x)\) also is, and note that for all events \(e'\) in \(z \setminus \pi_1(x)\), we have that \(e_1 \not \leqslant_{z \setminus \pi_1(x)} e' \).
	Since \(e' \in z \setminus \pi_1(x)\), \(e' \in z\) and \(e_1 \not \leqslant_z e'\), we have that \(z\) is a maximal configuration in \(\iconf_1\) where \(e_1\) is not maximal: a contradiction. Hence, we know that \(\pi_1(x)\) is maximal in \(\iconf_1\).

	Now, we prove that \(\pi_1(e) \neq \star\).
	For the sake of contradiction, suppose that \(\pi_1(e) = \star\).
	Then, observe that the underlying configuration structure \(\functoric (\iconf_1 \mid \iconf_2)\) also have this configuration \(x\) with \((e_1, \star) <_x e\).
	Without loss of generality, we can assume that \(e\) is \emph{an immediate cause} of \((e_1, \star)\)~\cite[Definition 20]{Aubert2016jlamp}, that is, that there is no \(e''\) in \(x\) such that \((e_1, \star) <_x e'' <_x e\).
	We can now use a small proposition connecting events in products with the order in their original configuration~\cite[Proposition 3]{Aubert2016jlamp} to get that it must be the case that either \(\pi_1(e_1, \star) <_{\pi_1(x)} \pi_1(e)\) or \(\pi_2(e_1, \star) <_{\pi_2(x)} \pi_2(e)\).
	But since \(\pi_2(e_1, \star) = \star \notin \pi_2(x)\), and since we assumed \(\pi_1(e) = \star \notin \pi_1(x)\), both scenario are impossible, so we reached a contradiction, and it must be the case that \(\pi_1(e) \neq \star\).

	Since \(\pi_1(x)\) is maximal in \(\iconf_1\), and since \(\pi_1(e), e_1 \in \pi_1(x)\), we reached a contradiction: \(\pi_1(e)\) is not caused by \(e_1\), but that cannot be since \(e_1\) is not maximal.
	Hence, there is no such \(e \in x\), and \((e_1, \star)\) is maximal in \(\iconf_1 \mid \iconf_2\).
\end{proof}

We can now use Lemmas~\ref{lem:equivalence_encoding}--\ref{lem:prevent-max-event} to prove our main result for this section. We use an extra notation here, we write \(d < e\) if for every \(x \in C\) such that \(d, e \in x\), we have \(d <_x e\).

\transitionidentified*

\begin{proof}

	We can rephrase the lemma as follows:

	\begin{enumerate}
		\item For all transitions \(R\fwlts{i}{\alpha}S\), \(\encmp{R} \iso \encmp{S}\crestr{\{e\}}\) with \(e\) maximal in \(\encmp{S}\) and \(\iden_S(e) = i\). \label{lem:transition_identified-fw}
		\item For all transitions \(R\bwlts{i}{\alpha}S\), \(\encmp{S} \iso \encmp{R}\crestr{\{e\}}\) with \(e\) maximal in \(\encmp{R}\) and \(\iden_R(e) = i\). \label{lem:transition_identified-bw}
		\item For any maximal event \(e\) in \(\encmp{R}\), there is %
		      a transition \(R \bwlts{\iden_R(e)}{\labl_R(e)} S\) with \(\encmp{S} \iso \encmp{R}\crestr{\{e\}}\). \label{lem:transition_identified-bw-bis}
	\end{enumerate}

	We prove the three items separately.

	\subparagraph*{\autoref{lem:transition_identified-fw}}
	We proceed by case on \(\alpha\) in the transition \(R\fwlts{i}{\alpha}S\):
	\begin{description}
		\item[If \(\alpha = a\)]
		      For the transition \(R\fwlts{i}{\alpha}S\) to take place, it must be the case that \(R\) contains a thread \(T = m \rhd a.P + Q\) for some \(m\), \(P\) and \(Q\), and that \(T \fwlts{i}{\alpha}\mem{i, a, Q} . m \rhd P \).
		      Additionally, the par.\(_{\text{L}}\), par.\(_{\text{R}}\), res. and \(\congru\) rules of \autoref{fig:ltsrules}, gives us that it must be the case that
		      \begin{align*}
			      R \congru ((R_{n-1} \cdots((R_3 \mid ((R_1 \mid T) \mid R_2)\bs \namelist{b^1}) \mid R_4)\bs \namelist{b^2} \cdots ) \mid R_n) \bs \namelist{b^{m}}
		      \end{align*}
		      for some \(R_i\) any of which (and its corresponding \(\mid\) constructor) could be missing
		      and for some \(\namelist{b^j}\), any of which (along with their \(\bs\) constructor) could be missing as well.
		      Hence, the transition can be written as
		      \begin{align*}
			      R \fwlts{i}{a} \underbrace{((R_{n-1} \cdots((R_3 \mid ((R_1 \mid \mem{i, a, Q} . m \rhd P) \mid R_2)\bs \namelist{b^1}) \mid R_4)\bs \namelist{b^2} \cdots ) \mid R_n) \bs \namelist{b^{m}}}_{= S'}
		      \end{align*}
		      with \(S' \congru S\).
		      By \autoref{lem:postfix-max-event}, we know that there is a maximal event \(e_1\) in \(\encmp{\mem{i, a, Q} . m \rhd P}\) that has for identifier \(i\).
		      We show that this event can be \enquote{traced through} \(\encmp{S'}\), using four arguments:
		      \begin{itemize}
			      \item From the par.\(_{\text{L}}\) or par.\(_{\text{R}}\) rule in \autoref{fig:ltsrules} we have that \(i \notin \ids(R_1)\).
			            Using \autoref{lem:prevent-max-event}, it follows that there exists a maximal event \(e_2\) in \(\encmp{R_1 \mid \mem{i, a, Q} . m \rhd P}\) such that \(\pi_2(e_2) = e_1\), and from the definition of the parallel composition of \(\iconf\)-structures (\autoref{def:par-ics}) that its identifier is \(i\).
			      \item We can use the same reasoning to prove that there is a maximal event \(e_3\) in \(\encmp{(R_1 \mid \mem{i, a, Q} . m \rhd P) \mid R_2}\) such that \(\pi_1(e_3) = e_2\) and such that its identifier is \(i\).
			      \item Since \(\encmp{R_3 \mid ((R_1 \mid \mem{i, a, Q} . m \rhd P) \mid R_2)\bs \namelist{b^1}} = \encmp{R_3 \mid ((R_1 \mid \mem{i, a, Q} . m \rhd P) \mid R_2)}\), it is trivial that \(e_3\) is a maximal event with identifier \(i\) in it.
			      \item Using those three arguments repeatedly, and \enquote{skipping} them if a parallel composition or a restriction is \enquote{missing}, we can \enquote{trace} the maximal event with identifier \(i\) in \(\encmp{S'}\), that we write \(e\).
		      \end{itemize}
		      But we still need to find a maximal event whose identifier is \(i\) \emph{in \(\encmp{S}\)}.
		      Since \(S' \congru S\),
		      \autoref{lem:equivalence_encoding} gives us that \(\encmp{S'} \iso \encmp{S}\).
		      Let us write \(f\) this isomorphism, we know that it is such that \(f_L = \id\), and hence that we can use \autoref{rem:mor-in-catics-notation} to study only the part of \(f\) that maps events, that we also write \(f\).
		      Now, we want to prove that \(f(e)\) is maximal in \(\encmp{S}\).
		      To do so, let us use the \enquote{if} part of the upcoming \autoref{lem:lop_preserve_conf}\footnote{%
			      This part of the lemma is \emph{not} proved using the lemma we are currently proving, but the \enquote{and only if} part does, hence the choice of postponing it after this current lemma.} %
		      to obtain a \lop bijection between the maximal configurations of \(\encmp{S'}\) %
		      and the maximal configuration of \(\encmp{S}\), that we also write \(f\).
		      Now, let's suppose that \(f(e)\) is not maximal in \(\encmp{S}\).
		      Then, there exists a maximal configuration \(x_m\) in \(\encmp{S}\) and an event \(e'\) in \(x_m\) such that \(f(e) <_{x_m} e'\).
		      But since \(f\) is an isomorphism, and since \(f\) is order-preserving, %
		      we have that it must be the case that \(e <_{f^{-1}(x_m)} f^{-1} (e')\), which contradicts the maximality of \(e\) in \(S'\).%

		      Hence, \(f(e)\) is maximal in \(\encmp{S}\), and since by \autoref{lem:equivalence_encoding} \(f_{\iden} = \id\), \(\iden_S(e) = i\). %
		      It is obvious that this event is the only one that was added to the encoding of \(R\), so that \(\encmp{R} = \encmp{S}\crestr{\{e\}}\).

		\item[If \(\alpha = \tau\)]
		      Then it must be the case that \(R\) has two threads, \(T = m \rhd a.P + Q\) and \(T' = m' \rhd \bar{a}.P' + Q'\) for some \(m\), \(m'\), \(a\), \(P\), \(P'\), \(Q\), and \(Q'\), and that
		      \begin{align*}
			      R = (\cdots(R_3 \mid (((R_1 \mid T) \mid R_2)\bs \namelist{b^1} \mid R_4))\bs \namelist{b^2} \cdots \mid((R_1'\mid T')\mid R_2')\bs \namelist{c^1} \cdots \mid R_n) \bs \namelist{b^{m}}
		      \end{align*}

		      for some \(R_i, R_i'\) any of which (and its corresponding \(\mid\) constructor) could be missing
		      and for some \(\namelist{b^j}, \namelist{c^k}\), any of which (along with their \(\bs\) constructor) could be missing as well.
		      Then the transition becomes
		      \begin{align*}
			      R \fwlts{i}{\tau}
			       & (\cdots((R_3 \mid (R_1 \mid (\mem{i, a, Q} . m \rhd P) \mid R_2))\bs \namelist{b^1} \mid R_4)\bs \namelist{b^2} \cdots                    \\
			       & \hspace{2em} (R_1'\mid ((\mem{i, \bar{a}, Q'} . m' \rhd P') \mid R_2')\bs \namelist{c^1} \cdots \mid R_n) \bs \namelist{b^{m}} \congru S.
		      \end{align*}

		      We use the same reasoning as in the case above where \(\alpha = a\) to deduce that there is a maximal event in \(\encmp{(\mem{i, a, Q} . m \rhd P}\), let us name it \(e_1\), and one in \(\encmp{\mem{i, \bar{a}, Q'} . m' \rhd P'}\), let us name it \(e'_1\), that both have identifier \(i\).
		      Using the same argument as in the first case, we can \enquote{trace} in parallel \(e_1\) and \(e'_1\), until the \(\iconf\)-structures that hold their \enquote{descendants} are put in parallel: at that point, by the definition of parallel composition of \(\iconf\)-structures (\autoref{def:par-ics}), it should be clear that a single maximal event resulting from their composition will emerge, that it will be labeled \(\tau\) and have identifier \(i\).
		      Finally, using again the same argument as in the first case, this event resulting from the synchronization of our two maximal events will still be maximal in \(\encmp{S}\), and hence we can conclude that there exists a maximal event in \(\encmp{S}\) labeled \(\tau\) whose identifier is \(i\).
	\end{description}

	\subparagraph*{\autoref{lem:transition_identified-bw}}
	If \(R\bwlts{i}{\alpha}S\), then by the Loop Lemma~\cite[Lemma 6]{Danos2004}, \(S\fwlts{i}{\alpha}R\).
	By \autoref{lem:transition_identified-fw}, we have that \(\encmp{S} \iso \encmp{R}\crestr{\{e\}}\) with \(e\) maximal in \(\encmp{R}\) and \(\iden_R(e) = i\), which is what we wanted to show.

	\subparagraph*{\autoref{lem:transition_identified-bw-bis}}
	For any reachable process \(R\), from~\autoref{def:coh} and the remark below it~\cite[Lemma 10]{Danos2004}, we have that there exists a \emph{forward-only} trace \(\orig{R}\tfbwlts R\). We consider \withoutlog the following trace: \(\emptymem \rhd \orig{R} = R_0\fwlts{1}{\alpha_1} R_1\cdots \fwlts{k}{\alpha_k} R_{k} = R\), for some \(k\), \(\alpha_1, \hdots, \alpha_k\).

	Using \autoref{lem:transition_identified-fw}, we have that %
	\(\encmp{R_j}\crestr{\{e_j\}} = \encmp{R_{j-1}}\), for \(j\leqslant k\) and \(\iden(e_j)=j\). Therefore we can construct a bijection \(h\) between events \(e_j\) in \(\encmp{R}\) and transitions \(t_j : R_{j-1} \fwlts{j}{\alpha_j} R_{j}\) with \(\iden(e_j) = j\).

	Let \(e\) be a maximal event in \(\encmp{R}\) and let \(h(e) = R_{j-1}\fwlts{j}{\alpha_j} R_j\), where
	\begin{align*}
		R_0 \fwlts{1}{\alpha_1} R_1\cdots R_{j-1}\fwlts{j}{\alpha_j} R_j\fwlts{j+1}{\alpha_{j+1}} R_{j+1}\cdots R_{k-1}\fwlts{k}{\alpha_k} R_{k} = R
	\end{align*}

	As \(e\) is maximal, by \autoref{def:maximal-event}, there exists no event \(e'\) in \(\encmp{R}\) such that \(e < e'\).
	From %
	\(\encmp{R_{i+1}}\crestr{\{e_{i+1}\}} = \encmp{R_i}\), for \(i\leqslant k\), we have that \(\encmp{R_i} \subseteq \encmp{R_{i+1}}\) and by induction, \(\encmp{R_i} \subseteq \encmp{R_{k}}\), for all \(i\leqslant k\).
	Then \(\encmp{R_{j+i}}\subseteq \encmp{R_k}\) and we have that \(e_i\not\leqslant e\), for all \(j+1\leqslant i\leqslant k\). It implies that \(e_i\) is concurrent with \(e\), for all \(j+1\leqslant i\leqslant k\).

	We use now concurrency in traces, trace equivalences and the \(m_{R / S}\) notations from \autoref{sec:trace}.
	Consider two consecutive transitions \(t_{j} = R_{j}\fwlts{j+1}{\alpha_{j+1}} R_{j+1}\) and \(t_{j+1} = R_{j+1}\fwlts{j + 2}{\alpha_{j+2}} R_{j+2}\).
	Using the bijection \(h\) defined above between transitions and events, we have that there exists two event \(e_{j}, e_{j+1}\) in \(\encmp{R_{j+2}}\) such that \(h(e_{j}) = t_{j}\) and \(h(e_{j+1}) = t_{j+1}\).
	Suppose that \(e_{j}\) is concurrent to \(e_{j+1}\).
	By an induction on the structure of memories in \(\encmp{R_{j+2}}\), we have that \(m_{R_{j}/R_{j+1}} \cap m_{R_{j+1}/R_{j+2}}=\emptyset\) which implies that there exists a transition \(t_{j}' = R_{j}\fwlts{j+2}{\alpha_{j+2}} R_{j+2}\) such that \(t_{j}\) is concurrent with \(t_{j}'\).

	Recall that \(e_i\) is concurrent with \(e\), for all \(j+1\leqslant i\leqslant k\) in our trace.
	We can now use the trace equivalence~\cite[Definition 9]{Danos2004} reminded in \autoref{sec:trace}, and use the previous remark repeatedly to re-organize our original trace as follows:

	\begin{align*}
		R_0 \fwlts{1}{\alpha_1} R_1\cdots R_{j-1} \fwlts{j+1}{\alpha_{j+1}} R'_{j+1}\cdots R'_{j-1}\fwlts{k}{\alpha_k} R'_{k} \fwlts{j}{\alpha_j} R
	\end{align*}
	The core idea being that we can \enquote{postpone} the transition that will create the event \(e\) in the encoding of \(R\), by flipping it with the other transitions, and have it become the last transition that leads to the same \(R\).

	We can now use the Loop Lemma~\cite[Lemma 6]{Danos2004} to \enquote{reverse} the last transition and get that \(R \bwlts{j}{\alpha_j} R'_{k}\).
	Using \autoref{lem:transition_identified-bw} on this transition, we have that \(\encmp{R'_k} = \encmp{R}\crestr{\{e\}}\).
\end{proof}

\subsection{Proof for \autoref{sec:HHPB-CCS}}
\label{subsec:collapse-proof}

\thequivalencebf*

\begin{proof}
	Let \(R_1\), \(R_2\) be two reversible processes without auto-concurrency.
	We want to show that
	\begin{enumerate}
		\item \((R_1, R_2, f) \in \rel'\) for \(\rel'\) a \BF implies that \((R_1, R_2) \in \rel\) for \(\rel\) a \SBF. %
		      \label{prof:hhpb_ident2}
		\item \((R_1, R_2) \in \rel\) for \(\rel\) a \SBF %
		      implies that there exists \(f\) a bijection between \(\ids(R_1)\) and \(\ids(R_2)\) such that \((R_1, R_2, f) \in \rel'\) is a \BF. \label{prof:hhpb_ident1}
	\end{enumerate}

	For \autoref{prof:hhpb_ident2}, there is nothing to prove, since we are dropping a requirement, the existence of the bijection.

	For \autoref{prof:hhpb_ident1}, we first note that there exists forward-only traces from \(\emptymem \rhd \orig{R_1}\) to \(R_1\)~\cite[Lemma 10]{Danos2004}.
	We pick one, \(\theta_1\), and construct the bijection \(f\) between \(\ids(R_1)\) and \(\ids(R_2)\) using it.
	We start by letting \((\emptyset \vartriangleright \orig{R_1}, \emptyset \vartriangleright \orig{R_2}, \emptyset) \in \rel'\).
	Then, assuming \(\theta_1\) starts with \(\emptyset \vartriangleright \orig{R_1} \fwlts{i}{\alpha} R_1'\), since \((R_1, R_2) \in \rel\), we know there are \(R_2'\) and \(j\) such that \(\emptyset \vartriangleright \orig{R_2} \fwlts{j}{\alpha} R_2'\) and \((R_1', R_2') \in \rel\), and we let \((R_1', R_2', \{i \to j\}) \in \rel'\).
	We iterate this construction until we obtain a bijection \(f\) and let \((R_1, R_2, f) \in \rel'\).

	Now, we need to show that \((R_1, R_2, f) \in \rel'\) is a valid relation, and for this we must show that it can accommodate the forward and backward transitions.
	That the bijection \(f\) can be extended in case of forward transitions from \(R_1\) or \(R_2\) is obvious, using \(\rel\).
	The more difficult part of the proof concerns the backward transition, which requires to show that \(f\) is \enquote{right}.
	Indeed, if \(R_1\) or \(R_2\) does a backward transition to a term \(S\), and if \(S\) is paired with \(S'\) in \(\rel\), then we need to prove that the identifiers of the backward transitions leading to \(S\) and \(S'\) are %
	\enquote{matched}, \ie, in bijection in \(f\).

	Stated formally, we want to show that for any \((R_1, R_2) \in \rel\) with \(f\) constructed as above, then for all \(R_1\bwlts{i}{\alpha}S_1\) and \(R_2\bwlts{j}{\alpha}S_2\) such that \((S_1, S_2) \in \rel\), \(\{i\mapsto j\} \in f\).
	We show this by contradiction:
	suppose that we have two pairs of indices \(\{i\mapsto j\}\) and \(\{i'\to j'\}\) in \(f\), and that \(R_1\bwlts{i'}{\alpha}S_1\) and \(R_2\bwlts{j}{\alpha}S_2\) such that \((S_1, S_2) \in \rel\).

	First, observe that all four indices , \(i\), \(i'\), \(j\) and \(j'\), are associated to the same label \(\alpha\): due to the way \(f\) was constructed, following \(\rel\), it cannot be the case that two transitions with different labels have their identifiers paired.
	Secondly, since \(R_2\bwlts{j}{\alpha}S_2\) and \(\{i \to j\}\) is in \(f\), it must be the case that there is a transition in \(\theta_1\) with the identifier \(i\) by construction of \(f\). Let us denote \(t_i:S_1'\fwlts{i}{\alpha}R_1'\) and \(t_{i'}:S_1\fwlts{i'}{\alpha}R_1\).

	Now, let us put the remarks on trace of \autoref{sec:trace} to good use.
	Observe that \(i\) and \(i'\) cannot be concurrent in \(\theta_1\).
	If they were, then in the encoding of \(R_1\), there would be two events--the one introduced by \(t_i'\) and the one introduced by \(t_j\)--that are concurrent, with the same label, but that is not possible since \(R_1\) is without auto-concurrency.
	Since two transitions are either concurrent or causal, it follows that either \(i<i'\) or \(i'<i\)\footnote{Remember that we write \(i < k\) for \(t_i < t_k\) .}.

	\begin{itemize}
		\item As the cause cannot backtrack before the effect it follows that the case \(i'<i\) is not possible.
		\item Then \(i<i'\).
		      We now reason by cases on the order of the traces \(t_j\) and \(t_{j'}\): either \(j<j'\) or \(j'<j\).
		      \begin{itemize}
			      \item If \(j<j'\) then \(R_1\) must have already backtracked on \(j'\), which implies that the there is no \(j'\) in \(f\) as we assumed.
			      \item The case \(j'<j\) is not possible. Note that \(i<i'\) implies an order in which the pair of indices are added to the bijection, in this case the pair \(\{i\mapsto j\}\) occurs before \(\{i'\to j'\}\) and therefore it cannot be that \(j'<j\).
		      \end{itemize}
	\end{itemize}

	From this reasoning, we get that for any \(R_1\bwlts{i}{\alpha}S_1\) and \(R_2\bwlts{j}{\alpha}S_2\) such that \((S_1, S_2) \in \rel\), \(\{i\mapsto j\} \in f\).
	We still need to show that \((S_1, S_2, f \setminus \{i \to j\}) \in \rel'\) is a \enquote{valid} relation.
	We only need to show that \(f \setminus \{i \to j\}\) matches the \enquote{right} indices in case of a backward transition, but we can simply iterate the previous reasoning until we reach \((\emptyset \vartriangleright \orig{R_1}, \emptyset \vartriangleright \orig{R_2}, \emptyset) \in \rel'\).
	This concludes this direction of the proof.
\end{proof}

\subsection{Lemmas and Proofs for %
	\autoref{sec:lifting}}
\label{subsec:proof-main}

Our goal here is to prove Theorems \ref{th:hpb_rccs} and \ref{thm:def_hpb_coincide}, which give as an immediate corollary our main result (\autoref{cor:previousresult}) and an interesting remark (\autoref{cor:fw-with-id}).
We first (\autoref{ssec:csquences}) identify isomorphisms of \(\iconf\)-structures with \lop functions (\autoref{lem:lop_preserve_conf}) and then connect our previous formalism with the encoding of memories (\autoref{lem:enc_mem_2}).
Then, using this connection and the operational correspondence detailed in \autoref{sec:op-corr}, \autoref{ssec:main_proofs} concludes by proving our two main theorems.

Below, we abuse the notation by writing \(\enc{R} = (\enc{\orig{R}}, x_R)\), for the encoding of a reversible process into a pair made of the CCS encoding of its origin and a particular configuration in it, that we call the address of \(R\)~\cite{Aubert2016jlamp}.
The formal connection made in \autoref{lem:enc_mem_2} will justify this abuse, and the context should always make it clear which encoding we are referring to.

\subsubsection{On Connecting Formalisms}
\label{ssec:csquences}

We first establish a connection between isomorphisms (in Category theory) and \lop bijections (that are used to define \HorHPB).
Then, we construct a bridge between \(\encm{R}\) and our previous formalism~\cite{Aubert2016jlamp}. %

\begin{restatable}{lemma}{loppreserveconf}
	\label{lem:lop_preserve_conf}
	Letting \(x_1^m\) and \(x_2^m\) be the unique maximal configurations in \(\encmp{R_1}\) and \(\encmp{R_2}\),
	there is an isomorphism \(f\) between \(\encmp{R_1}\) and \(\encmp{R_2}\) with \(f_L = \id\)
	iff there exists a \lop bijection between \(x_1^m\) and \(x_2^m\).
\end{restatable}

\begin{proof}
	Let \(f:\encmp{R_1}\to\encmp{R_2}\) be an isomorphism: since, by hypothesis, \(f_L = \id\), we can use \autoref{rem:mor-in-catics-notation} to take \(f\) to be fully determined by its function between events \(f_E\), that we also write \(f\).
	Then \(f:x_1^m\to x_2^m\) is by definition a label-preserving bijection, as instantiating \autoref{def:cat_lcs} with \(f_L = \id\) precisely gives the condition of \autoref{def:label-order-functions}.
	For two events \(e,e'\) in \(\encmp{R_1}\), if \(e<_{x_1^m} e'\) it follows that for all \(x_1\) in \(\encmp{R_1}\) such that \(e'\in x_1\) then \(e\in x_1\). From \(f\) an isomorphism we have that \(f^{-1}:\encmp{R_2}\to\encmp{R_1}\) is a well defined morphism. As \(f^{-1}\) preserves configurations it follows that for all \(x_2\) in \(\encmp{R_2}\), there exists \(x_1\) in \(\encmp{R_1}\) such that \(f(x_1) = x_2\). Then for all \(x_2\), such that \(f(e')\in x_2\) it implies that \(e'\in f^{-1}(x_2)\) and \(e\in f^{-1}(x_2)\) and finally, \(f(e)\in x_2\). Then \(f(e)<_{x_2^m} f(e')\), as for all \(x_2\), \(x_2\subseteq x_2^m\), since \(\encmp{R_2}\) is a poset (\autoref{lem:prefix_enc}).
	Similarly we proceed for the concurrent events, and obtain that \(f\) is an order-preserving bijection on top of being label-preserving, hence that it is \lop.

	For the reverse direction, let \(f:x_1^m\to x_2^m\) be a \lop bijection. All events in \(\encmp{R_1}\) are present in \(x_1^m\), since it is the maximal element in \(\encmp{R_1}\), and therefore \(f:\encmp{R_1}\to\encmp{R_2}\) is a bijection on events.
	As \(f\) is label-preserving, it is the identity on the labels, and we can study only this mapping on events thanks to \autoref{rem:mor-in-catics-notation}.
	Remains to show that \(f\) preserves configurations, that is, for all \(x_1\subseteq x_1^m\), \(f(x_1) =\{f(e) \setst e \in x_1\}\) is a configuration in \(\encmp{R_2}\). In other words, \(f(x_1)=x_2\) must satisfies the properties of~\autoref{def:conf_str}.
	\begin{itemize}
		\item \ref{Finiteness} follows from the fact that all configurations in the encoding of an RCCS memory are finite.
		\item For \ref{CoincidenceFreeness} let us first note that there is a unique and finite maximal configuration \(x_2^m\) in which all distinct events are either causal or concurrent.
		      Let us consider \(e_2\neq e_2'\) two events in \(x_2\). Since \(e_2\neq e_2'\) then either
		      \begin{enumerate*}
			      \item \(e_2\leqslant_{x_2^m} e_2'\) (or \(e_2\leqslant_{x_2^m} e_2'\) but this is similar), or \label{lem:lop_preserve_conf-item1}
			      \item \(e_2\co_{x_2^m} e_2'\). \label{lem:lop_preserve_conf-item2}
		      \end{enumerate*}
		      As \(x_2^m\) is the unique maximal configuration, the relations above hold for \(x_2\) as well.
		      In the case \autoref{lem:lop_preserve_conf-item1}, as \(e_2\neq e_2'\) then \(e_2'\not\leqslant_{x_2} e_2\). We apply the definition of causality (\autoref{def:causality}) to obtain the configuration \(z\subseteq x_2\) where \(e_2\in z\) and \(e_2'\notin z\) as required by \ref{CoincidenceFreeness}.
		      In the case \autoref{lem:lop_preserve_conf-item2}, we use the definition of concurrency, which implies that \(e_2\not\leqslant_{x_2} e_2'\) and \(e_2'\not\leqslant_{x_2} e_2\). Then there exists \(z\subseteq x_2\) such that \(e_2\in z\) and \(e_2'\notin z\) and moreover, there exists \(z'\subseteq x_2\) such that \(e_2'\in z\) and \(e_2\notin z\).
		\item To show \ref{FiniteCompleteness}, it suffices to note that there exists a unique maximal configuration \(x_2^m\) which can be an upper bound for the union of any subset of configurations in \(\encmp{R_2}\).
		\item To show \ref{Stability} we have to show that \(\forall x_2,x_2'\in C_2\) such that \(\exists x_1,x_1'\in C_1\), with \(x_1 = f^{-1}(x_2)\) and \(x_1' = f^{-1}(x_2')\), \(x_2\cup x_2'\in C_2 \implies x_2\cap x_2' \in C_2\).
		      We will show that \(x_2\cap x_2' = f(x_1 \cap x_1')\) is a configuration in \(\encmp{R_2}\) by exploiting the causality and concurrency relations in \(x_1 \cap x_1'\).
		      As \(\encmp{R_1}\) and \(\encmp{R_2}\) are both posets by \autoref{lem:prefix_enc}, we can w.l.o.g consider \(x_1'\subseteq x_1\) and \(x_1 = x_1^{m}\).
		      We can assume, for \(i \in \{1,2\}\), that \(\encmp{R_i} \neq \confzero\) (otherwise the property would be trivial).
		      Hence, there is a backward transition \(R_i \bwlts{j}{\alpha} S_i\).
		      Now, by the second item of \autoref{lem:transition_identified} we have that there exists a maximal event \(e_1^1\) in \(x_1\) and from \(\encmp{R_1}\crestr{\{e_1^1\}}\subset \encmp{S_1}\), \(x_1{\setminus} e_1\) is a configuration in \(\encmp{R_1}\). We can re-use the same reasoning for \(x_1{\setminus} \{e_1\}\) in \(\encmp{S_1}\crestr{\{e_1^2\}}\) with \(e_1^2\) maximal and \(\encmp{R_1}\crestr{\{e_1^1, e_1^2\}}\subset\encmp{S_1}\crestr{\{e_1^1\}}\subset \encmp{S_1}\). We have then a sequence of events \(e_1^1, \cdots e_1^n\) such that \(x_1{\setminus} \{e_1,\cdots, e_n\} = x_1 \cap x_1'\).

		      Let us consider \(e_2\neq e_2'\) two events in \(x_2\) and as \(f\) is a bijection, let \(e_1\neq e_1'\) be two events in \(f^{-1}(x_2) = x_1\) such that \(f^{-1}(e_2)=e_1\) and \(f^{-1}(e_2')=e_1'\).
		      If \(e_2\neq e_2'\) then either \(e_2\leqslant_{x_2^m} e_2'\) or \(e_2\co_{x_2^m} e_2'\). It follows that either \(e_1\leqslant_{x_1^m} e_1'\) or \(e_1\co_{x_1^m} e_1'\), respectively, since \(f^{-1}\) is \lop as well. As \(x_1^m\), \(x_2^m\) are the unique maximal configurations the relations above hold for \(x_2\) and \(x_1\) as well.

		      We therefore have that \(f(e_1^1)\) is the maximal event in \(x_2\). Again we use the second item of \autoref{lem:transition_identified} and as above we obtain the sequence of events \(f(e_1^1), \cdots f(e_1^n)\) such that \(f(x_1){\setminus} f(e_1),\cdots, f(e_n) = f(x_1 \cap x_1')\) is a configuration in \(\encmp{R_2}\). \qedhere
	\end{itemize}
\end{proof}

For the reader familiar with event structures, a configuration \(x\) defines an event structure \((x, \leqslant_x, \labl)\).
The construction below mirrors the transformation from an event structure to a configuration structure~\cite{Winskel1995}.

\begin{definition}[Generation of a \(\iconf\)-structure from a configuration]
	\label{def:xdownarrow}
	Given \(\iconf = (E, C, L, \labl, I, \iden)\), for \(x \in C\), the \emph{\(\iconf\)-structure generated by \(x\)} is \(x{\downarrow} = (x, \{ y \setst y\in C, y\subseteq x\}, \{ a \setst \exists e \in x_{R}, \labl(e) = a\}, \labl\frestr{x}, \{ i \setst \exists e \in x_{R}, \iden(e) = i\}, \iden\frestr{x})\).
\end{definition}

\begin{restatable}{lemma}{encmem}
	\label{lem:enc_mem_2}
	The \(\iconf\)-structure \(\encmp{R}\) is isomorphic to \(x_{R}{\downarrow}\), where \(\enc{R} = (\enc{\orig{R}}, x_R)\) is the encoding previously defined~\cite{Aubert2016jlamp}.
\end{restatable}

\begin{proof}
	As \(R\) is reachable there exists a forward-only trace~\cite[Lemma 10]{Danos2004} \(\theta:\orig{R}\fwlts{i_1}{\alpha_1}\cdots\fwlts{i_n}{\alpha_n} R\). %
	We reason then by induction on the trace \(\theta\). As the base case is similar, we only treat the inductive case.

	By induction, we have that the result holds for \(R\), that is let \(h:\encm{R} \to x_R {\downarrow}\) be an isomorphism. As both \(\encm{R}\) and \(x_R {\downarrow}\) result from encoding the same process \(R\), we can consider w.l.o.g. that \(h_L = \id\) and \(h_{\iden}=\id\). We use an argument similar to \autoref{lem:lop_preserve_conf} to show that instead of an isomorphism, we can reason on \(h:x^m \to x_R\) as a \lop bijection, where \(x^m\) is the maximal configuration of \(\encm{R}\), and \(x_R\) is the maximal configuration in \(x_R {\downarrow}\).

	Let \(R\fwlts{i}{\alpha} S\) be an RCCS transition. We have to show that \(\encmp{S}\iso x_S{\downarrow}\), where \(\enc{S} = (\enc{\orig{R}}, x_S)\).
	By the operational correspondence between \(R\) and \((\enc{\orig{R}}, x_R)\)~\cite[Lemma 6]{Aubert2016jlamp} there exists an event \(e\) in \(\enc{\orig{R}}\) such that
	\[
		R\fwlts{i}{\alpha} S \implies (\enc{\orig{R}},x_R)\redl{e}(\enc{\orig{R}},x_R\cup\{e\}),\text{ with }\labl(e)=\alpha\text{ and }\enc{S} = (\enc{\orig{R}},x_R\cup\{e\}).
	\]
	By the operational correspondence between \(R\) and \(\encmp{R}\) (\autoref{lem:transition_identified}) there exists a memory event \(e_m\) such that \(\encmp{R} = \encmp{S}\crestr{\{e_m\}}\) with \(e_m\) maximal in \(\encmp{S}\), \(\labl_S(e_m)= \alpha\) and \(\iden_S(e_m) = i\). As there is only one maximal configuration in \(\encmp{S}\) and as \(\encmp{R} = \encmp{S}\crestr{\{e_m\}}\), it follows that \(x^m\cup\{e_m\}\) is the maximal configuration in \(\encmp{S}\).

	We have then to show that if \(h:x^m \to x_R\) is \lop bijection then so is \(h' = h \cup \{e_m \mapsto e\}:x^m \cup\{e_m\}\to x_R\cup \{e\}\). The label preserving part follows from \(h_L = \id\) and \(\labl_S(e_m)= \labl(e) =\alpha\). To show that \(h'\) is order preserving we exploit again the operational correspondences between the encoding of RCCS terms and RCCS, and between the encoding of memories and RCCS.
	As a corollary from \cite[Lemma 6]{Aubert2016jlamp} we get that transitions are concurrent in the forward-only trace \(\theta:\orig{R}\fwlts{i_1}{\alpha_1}\cdots\fwlts{i_n}{\alpha_n} R\fwlts{i}{\alpha} S\) iff the corresponding events are concurrent in \(\enc{S}\), that is in \(x_S\). Similarly, from \autoref{lem:transition_identified}, we get that transitions are concurrent in the forward-only trace \(\theta\) iff their corresponding events are concurrent in \(\encmp{S}\), and consequently in \(x^m \cup\{e_m\}\). Finally we can conclude that \(h'\) is order preserving.

	Lastly we can lift \(h'\) from a \lop bijection on configurations to an isomorphism on \(\iconf\)-structures by endowing the configurations with the identifier set \(\ids(S)\) and \(h'_{\iden}=\id\).
\end{proof}

\subsubsection{Proofs of Theorems \ref{th:hpb_rccs} and \ref{thm:def_hpb_coincide}}
\label{ssec:main_proofs}

\hpbrccs*
\begin{proof}Let us prove the \HHPB case, the other case being similar, and actually simpler.
	\begin{description}
		\item[\(\Rightarrow\)]
		      Let \(\rel_{\text{RCCS}}\) be a \HHPB between \(P_1\) and \(P_2\) (\autoref{def:hpb_rccs}).
		      We show that the following relation
		      \[
			      {\rel} =
			      \begin{multlined}[t]
				      \{(x_1,x_2,f) \setst x_1\in\enc{P_1}, x_2\in\enc{P_2}, \exists R_1, R_2\text{ \st \ }\orig{R_1}=P_1,\orig{R_2} = P_2,\\
				      (R_1,R_2,F)\in \rel_{\text{RCCS}} \text{ and }\enc{R_1}=(\enc{P_1},x_1), \enc{R_2}=(\enc{P_2},x_2), f = \functoric(F)\}
			      \end{multlined}
		      \]
		      is a \HHPB between \(\enc{P_1}\) and \(\enc{P_2}\).

		      We first show that for any tuple \((x_1,x_2,f)\in\rel\), \(f:x_1\to x_2\) is a \lop bijection.
		      For a tuple \((R_1,R_2,F)\in\rel_{\text{RCCS}}\) with \(F:\encmp{R_1}\to\encmp{R_2}\) an isomorphism, we get, by~\autoref{lem:lop_preserve_conf}, that we can instead consider \(F:x_1^m\to x_2^m\) to be a \lop bijection between the maximal configurations \(x_1^m,x_2^m\). The functor \(\functoric\) maps isomorphisms on identified structures to isomorphisms on configuration structures, by~\autoref{lem:functors}, and therefore \(f = \functoric(F):\functoric(x_1^m)\to\functoric(x_2^m)\).

		      Moreover, \(x_1\) and \(x_2\) are the maximal configurations in \(x_1{\downarrow}\) and \(x_2{\downarrow}\), by~\autoref{def:xdownarrow}. By~\autoref{lem:enc_mem_2}, \(\encmp{R_i} \iso x_i{\downarrow}\) which implies that \(\functoric(x_i^m) = x_i\), for \(i\in\{1,2\}\). Thus \(f = \functoric(F)\) is well defined and indeed, a \lop bijection.

		      The rest of the proof follows the same structure as in~\cite[Proposition 6]{Aubert2016jlamp}.
		      Note that \((\emptyset,\emptyset,\emptyset)\in\rel\): indeed \((\emptymem\rhd P_1,\emptymem\rhd P_2,\emptyset)\in \rel_{\text{RCCS}}\) and \(\enc{\emptymem\rhd P_i} = (\enc{P_i}, \emptyset)\), for \(i\in\{1,2\}\).

		      Let us suppose that \((x_1,x_2,f)\in\rel\) for \(\enc{R_i}=(\enc{P_i},x_i)\), for \(i\in\{1,2\}\) and \(f:x_1\to x_2\) a \lop bijection.

		      To show that \(\rel\) is a \HHPB we have to show that if \(x_1 \redl{e_1} y_1\) (or \(x_1\revredl{e_1} y_1\)) then there exists \(y_2\) such that \(x_2 \redl{e_1} y_2\) (or \(x_2\revredl{e_2} y_2\) respectively) and such that \((y_1, y_2, f')\in\rel\) for some \(f' = f\cup\{e_1\mapsto e_2\}\).

		      Let \(x_1\redl{e_1}y_1\), hence by definition, \(y_1 = x_1\cup\{e_1\}\). From the correspondence between RCCS and their encodings~\cite[Lemma 6]{Aubert2016jlamp}, it follows that \(R_1\fwlts{i}{\alpha} S_1\) such that \(\enc{S_1} = (\enc{P_1},y_1)\).

		      As \((R_1,R_2,F)\in \rel_{\text{RCCS}}\) and as \(R_1\fwlts{i}{\alpha} S_1\), it follows that there exists a transition \(R_2\fwlts{j}{\alpha} S_2\) with \(F = F'\frestr{\encmp{R_1}}\) and \((S_1,S_2,F')\in\rel_{\text{RCCS}}\).

		      Again from the correspondence between \(R_2\) and \(\enc{R_2}\) we have that \(x_2\redl{e_2}y_2\) such that \(y_2=x_2\cup\{e_2\}\) and \(\enc{S_2} = (\enc{P_2}, y_2)\). Then we show that \((y_1,y_2,\functoric(F'))\in \rel\). The only missing argument is that \(\functoric(F') = f \cup\{e_1\mapsto e_2\}\). Note that, again using~\autoref{lem:lop_preserve_conf}, we can consider \(F:x_1^m\to x_2^m\) and \(F':y_1^m\to y_2^m\) to be \lop bijections on the maximal configurations of \(\encmp{R_1}, \encmp{R_2}\) and \(\encmp{S_1}, \encmp{S_2}\), respectively. From the definition of postfixing and parallel composition, in Definitions \ref{icat-op-def} and \ref{def:par-ics}, we have that \(y_1^m = x_1^m\cup\{e_1\}\) and \(y_2^m = x_2^m\cup\{e_2\}\).
		      Therefore \(F' = F\cup\{e_1\mapsto e_2\}\) by which we conclude.

		      We treat similarly the cases where \(x_2\) does a transition, or when the transitions are backwards.

		\item[\(\Leftarrow\)]
		      Let \(\rel_{\text{CONF}}\) be a \HHPB between \(\enc{P_1}\) and \(\enc{P_2}\).
		      We show that the following relation
		      \begin{align*}
			      {\rel} =
			      \begin{multlined}[t]
				      \{(R_1,R_2,F) \setst \orig{R_1}=P_1, \orig{R_2} = P_2, \enc{R_1}=(\enc{P_1},x_1),\enc{R_2}=(\enc{P_2},x_2),\\\text{ with }
				      (x_1,x_2,f)\in\rel_{\text{CONF}}, F=(f,f_{\iden}), \\
				      f_{\iden}(i) = j, \iden_1(e) = i, \iden_2(f(e)) = j, \text{for all }e\in x_1\}
			      \end{multlined}
		      \end{align*}
		      is a \HHPB between \(\emptyset\rhd P_1\) and \(\emptyset\rhd P_2\).

		      For \((R_1,R_2,F)\in\rel\), let us show that \(F\) is an isomorphism between \(\encmp{R_1}\) and \(\encmp{R_2}\). For \(f:x_1\to x_2\) a \lop bijection, and for two functions \(\iden_1:x_1\to I_1\) (from \(\encmp{R_1}\)) and \(\iden_2:x_2\to I_2\) (from \(\encmp{R_2}\)) then there exists a unique function \(f_{\iden}: I_1\to I_2\) such that \(f_{\iden}(\iden_1(e)) = \iden_2(f(e))\), for all \(e\in x_1\). This follows from \ref{eq:collision} in the definition of identified structures and from \(f\) being a bijection. We write then \(F = (f,f_{\iden}):(x_1\oplus\iden_1) \to (x_2\oplus\iden_2)\). Moreover, as in the first case, we derive that \(x_1, x_2\) are maximal in \(\encmp{R_1},\encmp{R_2}\), respectively. We use \autoref{lem:lop_preserve_conf} to conclude that \(F:\encmp{R_1}\to\encmp{R_2}\) is an isomorphism.

		      We have that \((\emptymem\rhd P_1,\emptymem\rhd P_2,\emptyset)\in \rel\) as \((\emptyset,\emptyset,\emptyset)\in\rel_{\text{CONF}}\) and \(\enc{\emptymem\rhd P_i} = (\enc{P_i}, \emptyset)\), for \(i\in \{1,2\}\).

		      We suppose now that \((R_1,R_2,F)\in\rel\), with \(F:\encmp{R_1}\to \encmp{R_2}\). It implies that \((x_1,x_2,f)\in\rel_{\text{CONF}}\) for \(\enc{R_i}=(\enc{P_i},x_i)\), \(i\in \{1,2\}\) and that \(F=(f,f_{\iden})\) with \(f:x_1\to x_2\) \lop.

		      To show that \(\rel\) is a \HHPB we have to show that if \(R_1\fwlts{i}{\alpha} S_1\) (or \(R_1\bwlts{i}{\alpha} S_1\)) then there exists \(S_2\) such that \(R_2\fwlts{j}{\alpha} S_2\) (or \(R_2\bwlts{j}{\alpha} S_2\) respectively) and such that \((S_1,S_2,F')\in\rel\) for some \(F'\).

		      Let \(R_1\fwlts{i}{\alpha} S_1\). We use again the correspondence between RCCS and their encodings~\cite[Lemma 6]{Aubert2016jlamp} from which we have that there exists \(e_1\) and \(y_1=x_1\cup\{e_1\}\) such that \(x_1\redl{e_1}y_1\) and \(\enc{S_1} = (\enc{P_1},y_1)\). As \((x_1,x_2,f)\in\rel_{\text{CONF}}\) it implies that there exists \(e_2\), \(y_2\) and \(f'=f\cup\{e_1\mapsto e_2\}\) such that \(x_2\redl{e_2}y_2\) and \((y_1,y_2,f')\in \rel_{\text{CONF}}\).

		      Again, from the correspondence between RCCS and configuration structures we have that, from \(x_2\redl{e_2}y_2\), there exists \(S_2\) such that \(R_2\fwlts{j}{\alpha} S_2\) with \(\enc{S_2} = (\enc{P_2}, y_2)\). From \(f'=f\cup\{e_1\mapsto e_2\}\) it easily follows that \(F'=(f', f'_{\iden}) = (f\cup\{e_1\mapsto e_2\}, f_{\iden}\cup\{i\mapsto j\})\).
		      We conclude therefore that \((S_1,S_2,F')\in \rel\).

		      Similarly we show the cases where \(R_1\) does a backward transition, or if \(R_2\) does a forward or backward transition. \qedhere
	\end{description}
\end{proof}

\bfhpbcoincide*

\begin{proof}

	\begin{description}
		\item[\(\Rightarrow\)]
		      \(P_1, P_2\) be two processes and let \(\rel\) be a \BF relation between them as in \autoref{def:hhpb_ident}. Then for any \((R_1,R_2,f) \in \mathcal{R}\), \(R_1,R_2\) are two RCCS processes and \(f:\ids(R_1)\to\ids(R_2)\) is a bijection on their identifiers.
		      Letting \(\encmp{R_i} = (E_i,C_i,L_i,\labl_i, I_i,\iden_i)\), for \(i\in\{1,2\}\), we show that the relation
		      \[
			      \mathcal{S} =\big\{(R_1,R_2,F) \setst (R_1,R_2,f)\in\rel, F(e_1) = e_2 \iff f(\iden_1(e_1)) = \iden_2(e_2) \big\}
		      \]
		      is a \HHPB relation as in~\autoref{def:hpb_rccs}.

		      Let us start by observing that \(F = (F_E, F_L, F_C, F_{\iden})\) is correctly defined as an isomorphism between \(\encmp{R_1}\) and \(\encmp{R_2}\) by our condition.
		      We start by letting \(F_{\iden} = f\), and observe that by \autoref{lem:id_unique}, no two events can have the same identifier in \(\encmp{R_1}\), and similarly for \(\encmp{R_2}\): hence having the \(F_{\iden}\) component is enough to define the \(F_E\) component, and they are both isomorphisms.
		      By \autoref{def:hhpb_ident}, no two identifiers can be paired by \(f\) unless they have the same label: hence, \(\labl_1(e) = \labl_2(F_E(e))\) iff \(F_L(\labl_1(e)) = \labl_2(F_E(e))\), which forces \(F_L = \id\).
		      We see that \(F\) respects the conditions of \autoref{def:cat_lcs}, and that by \autoref{rem:mor-in-catics-notation}, we can simply write \(F\) for \(F_E\) in the following.

		      We prove that \(\mathcal{S}\) is a \HHPB relation as in~\autoref{def:hpb_rccs} by induction on the processes \(R_1\) and \(R_2\), as any process reachable from \(P_1\) has to be in \(\rel\) (and in \(\mathcal{S}\)) and vice versa.
		      The base case \((\emptymem\rhd P_1,\emptymem\rhd P_2, \emptyset)\in \mathcal{S}\) is trivial.

		      Suppose that for \((R_1,R_2,f) \in\rel\), we have that \((R_1,R_2,F) \in\mathcal{S}\).
		      \begin{description}
			      \item[For forward transitions]
			            Now let \(R_1\fwlts{i}{\alpha}S_1\) for which we have that \(R_2\fwlts{j}{\alpha}S_2\) and \(g:\ids(S_1)\to\ids(S_2)\) is a bijection defined as \(f\) on \(\ids(R_1)\subset\ids(S_1)\) and \(g(i) = j\). Let \(e_1\in\encmp{S_1}\) such that \(\iden_1(e_1)=i\) and \(e_2\in\encmp{S_2}\) such that \(\iden_2(e_2)=j\). Then \(G = F\cup\{e_1\mapsto e_2\}\) is defined by \(g(\iden_1(e_1))=\iden_2(e_2)\), and thus \((S_1,S_2,G)\in\mathcal{S}\).

			            We show now that \(G\) is an isomorphism between \(\encmp{S_1}\) and \(\encmp{S_2}\). \(G\) preserves identifiers by definition; it preserves labels as \(\labl_1(e_1) = \labl_2(e_2) = \alpha\). Lastly we have to show that it preserves configurations. We can use~\autoref{lem:lop_preserve_conf} to show instead that \(G:x_1^m\to x_2^m\) is a \lop function, where \(x_1^m,x_2^m\) are the maximal configurations in \(\encmp{S_1}\) and \(\encmp{S_2}\), respectively. Also by~\autoref{lem:lop_preserve_conf} we have that \(F:y_1^m\to y_2^m\) is a \lop function on \(y_1^m, y_2^m\) maximal configurations in \(\encmp{R_1}\) and \(\encmp{R_2}\), respectively. From~\autoref{lem:transition_identified}, \(x_1^m{\setminus}\{e_1\}\in\encmp{R_1}\) and using~\autoref{lem:prefix_enc} we have that \(y_1^m = x_1^m\setminus\{e_1\}\). Therefore we have that \(G:x_1^m\to x_2^m\) is a \lop function on all events \(e\neq e_1\).

			            Let us suppose, by contradiction, that there exists \(e\) such that \(e_1\co_{x_1^m} e\) but that \(e_2\co_{x_2^m} F(e)\) does not hold. Take the maximal event \(e\) with such a property. Moreover, let \(F(e)=e' <_{x_2^m} e_2\) \withoutlog.

			            There exists at least one sequence of events \(e_1'\leqslant\cdots\leqslant e_n'\) such that \(e \leqslant_{x_1^m} e_1'\) and such that \(e_n'\) is maximal. For simplicity we suppose that there is only one such sequence (the general case uses the same reasoning). From~\autoref{lem:transition_identified}, we have that \(S_1\bwlts{\iden(e_n')}{\labl(e_n')}S_1'\). We have \(F(e)< F(e_1')<\cdots< F(e_n')\) and \(F(e_n')\) maximal, since \(F\) is \lop on all these events. As \((S_1,S_2,g)\in\rel\), we have also \((S_1',S_2',g')\in\rel\), where \(g'\) is defined as \(g\). We apply this reasoning until we reach the process \(T_1\) where \(e\) is maximal. Then there exists \(T_2\) such that \((T_1,T_2,g'')\) and \(g''\) is defined as \(g\) on the identifiers \(i\), \(j\) and \(\iden(e)\).

			            Then \(e\) is maximal however \(e'\) is not and we reach a contradiction: from~\autoref{lem:transition_identified} \(S_1\bwlts{\iden(e)}{\labl(e)}S_1'\) but \(S_2\) cannot backtrack on \(e'\). It implies then that \(G:x_1^m\to x_2^m\) is a \lop function on all events in \(x_1^m\).

			      \item[For backward transitions]
			            Now let \(R_1\bwlts{i}{\alpha}S_1\) for which we have that \(R_2\bwlts{j}{\alpha}S_2\) and \(g:\ids(S_1)\to\ids(S_2)\) is a bijection defined as \(f\) on \(\ids(S_1)\subset\ids(R_1)\), and \((S_1, S_2, g) \in \rel\). %
			            Then we have \((S_1,S_2,G)\in\mathcal{S}\) with \(G\) is defined by \(g(\iden_1(e_1))=\iden_2(e_2)\).
			            We only have to show that \(G\) is the restriction of \(F\) to \(\encmp{S_1}\) and a bijection.
			            The latter follows from the former, as \(F\) is itself a bijection, and that \(G = F \frestr{\encmp{S_1}}\) follows from \(g\) being a restriction of \(f\) and the way \(G\) and \(F\) are defined.

		      \end{description}
		\item[\(\Leftarrow\)]
		      Let \(P_1, P_2\) be two processes and let \(\rel\) be a \HHPB relation between them as defined in \autoref{def:hpb_rccs}. Then for any \(R_1,R_2\) two RCCS processes and \(f:\encmp{R_1} \to \encmp{R_2}\) an isomorphism such that \((R_1,R_2,f) \in\rel\), \(f\) is also a bijection on the event identifiers of the memories of \(R_1,R_2\). The relation
		      \[
			      \mathcal{S} =
			      \begin{multlined}[t]
				      \big\{(R_1,R_2,F) \setst (R_1,R_2,f)\in\rel\text{ and } \\
				      F(i) = j \iff f(e_1)=e_2, \iden_1(e_1) =i, \iden_2(e_2) = j \big\}
			      \end{multlined}
		      \]
		      is a \BF relation as defined in \autoref{def:hhpb_ident}. As above, we are using the operational correspondence of~\autoref{lem:transition_identified}. \qedhere
	\end{description}
\end{proof}

\begin{corollary}
	\label{cor:fw-with-id}
	The relation obtained by considering only (\ref{bf1}--\ref{bf2}) in the definition of \BF (\autoref{def:hhpb_ident}) is equal to \HPB on CCS terms (\autoref{def:hpb_rccs}).
\end{corollary}

\begin{proof}
	This is an immediate consequence of \autoref{thm:def_hpb_coincide}, as the proof does not mix forward and backward moves, it suffices to consider only the forward moves to obtain this result.
\end{proof}
\clearpage{}%

\end{document}